\documentclass[11pt]{article}
\pdfoutput=1
\usepackage[T1]{fontenc}
\usepackage{lmodern}
\usepackage[protrusion=true,expansion=true]{microtype}
\usepackage{amsmath,amssymb,amsfonts,amsthm}
\usepackage{subcaption}
\usepackage{graphicx}
\usepackage{fullpage}
\usepackage[backref=page]{hyperref}
\usepackage{color}
\usepackage{wrapfig}
\usepackage{tikz}
\usetikzlibrary{decorations.pathreplacing}
\usepackage{setspace}
\usepackage{algorithm}
\usepackage[noend]{algpseudocode}
\usepackage[framemethod=tikz]{mdframed}
\usepackage{xspace}
\usepackage{pgfplots}
\usepackage{framed}
\usepackage{thmtools}
\usepackage{thm-restate}
\usepackage{tabu}
\usepackage{fancyhdr}
\pgfplotsset{compat=1.5}

\newtheorem{theorem}{Theorem}[section]

\newtheorem{lemma}[theorem]{Lemma}

\newtheorem{definition}[theorem]{Definition}

\newtheorem{fact}[theorem]{Fact}

\newenvironment{proofof}[1]{\begin{trivlist} \item {\bf Proof
#1:~~}}
  {\qed\end{trivlist}}

\newcommand{\namedref}[2]{\hyperref[#2]{#1~\ref*{#2}}}
\newcommand{\thmlab}[1]{\label{thm:#1}}
\newcommand{\thmref}[1]{\namedref{Theorem}{thm:#1}}
\newcommand{\lemlab}[1]{\label{lem:#1}}
\newcommand{\lemref}[1]{\namedref{Lemma}{lem:#1}}

\newcommand{\seclab}[1]{\label{sec:#1}}
\newcommand{\secref}[1]{\namedref{Section}{sec:#1}}
\newcommand{\applab}[1]{\label{app:#1}}
\newcommand{\appref}[1]{\namedref{Appendix}{app:#1}}
\newcommand{\factlab}[1]{\label{fact:#1}}
\newcommand{\factref}[1]{\namedref{Fact}{fact:#1}}

\newcommand{\figlab}[1]{\label{fig:#1}}
\newcommand{\figref}[1]{\namedref{Figure}{fig:#1}}
\newcommand{\alglab}[1]{\label{alg:#1}}
\renewcommand{\algref}[1]{\namedref{Algorithm}{alg:#1}}

\newcommand{\deflab}[1]{\label{def:#1}}
\newcommand{\defref}[1]{\namedref{Definition}{def:#1}}

\def \DistAugEq    {\mdef{\mathsf{DistAugEq}}}

\newcommand{\PPr}[1]{\ensuremath{\mathbf{Pr}\left[#1\right]}}
\newcommand{\PPPr}[2]{\ensuremath{\underset{#1}{\mathbf{Pr}}\left[#2\right]}}
\newcommand{\Ex}[1]{\ensuremath{\mathbb{E}\left[#1\right]}}

\renewcommand{\O}[1]{\ensuremath{\mathcal{O}\left(#1\right)}}
\newcommand{\tO}[1]{\ensuremath{\tilde{\mathcal{O}}\left(#1\right)}}
\newcommand{\eps}{\varepsilon}

\def \Proj    {\mdef{\text{Proj}}}
\def \BuildNet    {\mdef{\textsc{BuildNet}}}
\def \dist    {\mdef{\text{dist}}}
\def \cost    {\mdef{\text{cost}}}
\def \OPT    {\mdef{\text{OPT}}}

\def \frakM    {\mdef{\mathfrak{M}}}

\def \calA    {\mdef{\mathcal{A}}}
\def \calB    {\mdef{\mathcal{B}}}

\def \calD    {\mdef{\mathcal{D}}}
\def \calE    {\mdef{\mathcal{E}}}
\def \calF    {\mdef{\mathcal{F}}}
\def \calG    {\mdef{\mathcal{G}}}
\def \calH    {\mdef{\mathcal{H}}}
\def \calI    {\mdef{\mathcal{I}}}

\def \calM    {\mdef{\mathcal{M}}}

\def \calP    {\mdef{\mathcal{P}}}

\def \calS    {\mdef{\mathcal{S}}}

\newcommand{\mdef}[1]{{\ensuremath{#1}}\xspace}  
\DeclareMathOperator*{\median}{median}
\DeclareMathOperator*{\argmin}{argmin}

\DeclareMathOperator*{\polylog}{polylog}
\DeclareMathOperator*{\poly}{poly}
\DeclareMathOperator*{\supp}{supp}
\DeclareMathOperator*{\EMD}{EMD}
\DeclareMathOperator*{\WASSZZ}{WASSZZ}
\DeclareMathOperator*{\WASSD}{WASSD}

\DeclareMathOperator*{\Cost}{cost}
\DeclareMathOperator*{\Var}{Var}

\newcommand{\ignore}[1]{}

\newif\ifnotes\notestrue 
\ifnotes
\newcommand{\samson}[1]{\textcolor{blue}{{\bf (Samson:} {#1}{\bf ) }} \marginpar{\tiny\bf
             \begin{minipage}[t]{0.5in}
               \raggedright S:
            \end{minipage}}}
\newcommand{\vincent}[1]{\textcolor{purple}{{\bf (Vincent:} {#1}{\bf ) }} \marginpar{\tiny\bf
             \begin{minipage}[t]{0.5in}
               \raggedright V:
            \end{minipage}}} 
\else
\newcommand{\samson}[1]{}
\newcommand{\vincent}[1]{}
\fi

\makeatletter
\renewcommand*{\@fnsymbol}[1]{\textcolor{mahogany}{\ensuremath{\ifcase#1\or *\or \dagger\or \ddagger\or
 \mathsection\or \triangledown\or \mathparagraph\or \|\or **\or \dagger\dagger
   \or \ddagger\ddagger \else\@ctrerr\fi}}}
\makeatother

\providecommand{\email}[1]{\href{mailto:#1}{\nolinkurl{#1}\xspace}}

\definecolor{mahogany}{rgb}{0.75, 0.25, 0.0}
\definecolor{darkblue}{rgb}{0.0, 0.0, 0.55}
\definecolor{darkpastelgreen}{rgb}{0.01, 0.75, 0.24}
\definecolor{darkgreen}{rgb}{0.0, 0.2, 0.13}
\definecolor{darkgoldenrod}{rgb}{0.72, 0.53, 0.04}
\definecolor{darkred}{rgb}{0.55, 0.0, 0.0}
\definecolor{forestgreenweb}{rgb}{0.13, 0.55, 0.13}
\definecolor{greencss}{rgb}{0.0, 0.5, 0.0}
\definecolor{bleudefrance}{rgb}{0.19, 0.55, 0.91}

\hypersetup{
     colorlinks   = true,
     citecolor    = mahogany,
	 linkcolor	  = darkpastelgreen
}

\fancypagestyle{pg}
{
\lhead{}
\rhead{}
\cfoot{--\ \thepage\ --}

}
\begin{document}
\pagestyle{pg}

\allowdisplaybreaks


\title{Streaming Euclidean $k$-median and $k$-means with $o(\log n)$ Space}
\author{
Vincent Cohen-Addad\thanks{Google Research. E-mail: \email{vcohenad@gmail.com}.} 
\and
David P. Woodruff\thanks{Carnegie Mellon University. E-mail: \email{dwoodruf@andrew.cmu.edu}. 
Supported in part by a Simons Investigator Award.
}
\and
Samson Zhou\thanks{Texas A\&M University. E-mail: \email{samsonzhou@gmail.com}. 
Work done in part while at Rice University and UC Berkeley. 
}
}
\date{\today}

\maketitle

\begin{abstract}
We consider the classic Euclidean $k$-median and $k$-means objective on data streams, where the goal is to provide a $(1+\varepsilon)$-approximation to the optimal $k$-median or $k$-means solution, while using as little memory as possible. Over the last 20 years, clustering in data streams has received a tremendous amount of attention and has been the test-bed for a large variety of new techniques, including coresets, the merge-and-reduce framework, bicriteria approximation, sensitivity sampling, and so on. Despite this intense effort to obtain smaller sketches for these problems, all known techniques require storing at least $\Omega(\log(n\Delta))$ words of memory, where $n$ is the size of the input and $\Delta$ is the aspect ratio. A natural question is if one can beat this logarithmic dependence on $n$ and $\Delta$. In this paper, we break this barrier by first giving an insertion-only streaming algorithm that achieves a $(1+\varepsilon)$-approximation to the more general $(k,z)$-clustering problem, using $\tilde{\mathcal{O}}\left(\frac{dk}{\varepsilon^2}\right)\cdot(2^{z\log z})\cdot\min\left(\frac{1}{\varepsilon^z},k\right)\cdot\text{poly}(\log\log(n\Delta))$ words of memory. Our techniques can also be used to achieve two-pass algorithms for $k$-median and $k$-means clustering on dynamic streams using $\tilde{\mathcal{O}}\left(\frac{1}{\varepsilon^2}\right)\cdot\text{poly}(d,k,\log\log(n\Delta))$ words of memory. 
\end{abstract}


\section{Introduction}
Clustering is the problem of partitioning an input dataset to identify and extract important underlying structural information, so that ideally, points in the same cluster have similar properties and points in different clusters have dissimilar properties. 
Various formulations of clustering have become important cornerstones in combinatorial optimization, computational geometry, computer vision, data science, and machine learning, while classic formulations of clustering such as the $k$-median and $k$-means problems have been studied since the 1950s~\cite{steinhaus1956division,macqueen1967classification}. 

More generally, $k$-median and $k$-means are specific parameterizations of the $(k,z)$-clustering problem, which is one of the most commonly studied center-based clustering objectives, where the input data elements belong to a metric space, every cluster is induced by a corresponding center point, and the cost of each data element is a fixed function of the distance between the data point and the corresponding cluster center. 
Formally, the input is a set $X$ of $n$ points $x_1,\ldots,x_n$ along with a distance function $\dist$, a cluster parameter $k>0$, and a positive integer exponent $z>0$. The $(k,z)$-clustering objective is defined to be
\[\min_{C\subset\mathbb{R}^d, |C|=k}\sum_{i=1}^n\min_{c\in C}\dist(x_i,c)^z.\]
When $z=1$ and $z=2$, the problem is known as $k$-median clustering and $k$-means clustering, respectively. 

\paragraph{Clustering in the streaming model.} 
As modern datasets have significantly increased in size, often consisting of hundreds of millions of points, attention has shifted to large-scale computational models that require neither the storage of the dataset nor multiple passes over the data, both of which could require expensive memory or runtime usage.  
Perhaps the simplest of these models is the insertion-only streaming model, where in our setting, the points $x_1,\ldots,x_n$ of $X$ arrive sequentially, and the goal is to output a clustering whose cost is a  $(1+\eps)$-multiplicative approximation of the optimal clustering of $X$ while using space sublinear in $n$, for some input accuracy parameter $\eps>0$.  
Although it is not possible to explicitly output the label for each point in the dataset using sublinear space, streaming algorithms often instead implicitly define the clusters by outputting the center of each cluster, so that each point is assigned to the cluster with the closest center. 

Most streaming algorithms for $(k,z)$-clustering follow the same template. 
Over the course of the stream, they first construct a coreset, which is a small weighted subset of points from the original dataset (see \defref{def:coreset} for a formal definition).  
After the stream terminates, they then compute an optimal or near-optimal $(k,z)$-clustering of the coreset to produce the cluster centers. 

In an offline setting, where the input dataset $X$ is fully accessible at all times, the best known coreset constructions sample and store $\tO{\frac{k}{\eps^2}}\cdot\min\left(k,\frac{1}{\eps^z}\right)$\footnote{Here we use $\tO{f\left(n,d,k,\Delta,\frac{1}{\eps}\right)}$ to denote $\O{f\left(n,d,k,\Delta,\frac{1}{\eps}\right)}\cdot\polylog\left(f\left(n,d,k,\Delta,\frac{1}{\eps}\right)\right)$.} weighted points of the input dataset~\cite{Cohen-AddadLSS22}. 
Although the sampled points could have weight $\poly(n,k,d,\Delta)$, where $\Delta$ is the aspect ratio, these weights can be represented by $\O{1}$ words of space, where a single word is $\Theta(\log(nd\Delta))$ bits of space, i.e., the space necessary to represent a single coordinate of an input point along with the corresponding weight. 
Thus it follows that there exists a succinct representation of the input dataset with size \emph{independent} of the number $n$ of input points.  

Therefore, a natural question to ask is:
\begin{quote}
\emph{Is it possible to perform approximate $(k,z)$-clustering in the streaming model using $o_{k,d,\eps}(\log n)$ words of memory?}
\end{quote}
More specifically, the best offline coreset constructions use $\tO{\frac{kd}{\eps^2}\min\left(k,\frac{1}{\eps^z}\right)\log(n\Delta)}$ bits of space to store the weighted subset of the input points. 
\begin{quote}
\emph{Can $(k,z)$-clustering be performed in the streaming model with the same space complexity?}
\end{quote}

At first glance, the answer might seem unlikely. 
The simplest and perhaps most standard approach for insertion-only streams is the merge-and-reduce framework, which first partitions the stream into a number of blocks, creating a coreset on-the-fly for each block. 
The framework can be viewed as a binary tree, where each node contains a coreset of the merger of the coresets contained at the children nodes, and the root node is a coreset for the entire stream. 
However, the main pitfall is that because each step incurs distortion through the construction of a new coreset, the coresets for each block require accuracy $\left(1+\O{\frac{\eps}{\log n}}\right)$, which results in space $\tO{\frac{kd}{\eps^2}\min\left(k,\frac{1}{\eps^z}\right)\log(n\Delta)\log^3 n}$. 
Moreover, all known streaming algorithms implicitly leverage analysis that requires correctness at all times of the stream, thereby incurring polylogarithmic dependencies in the number $n$ of input points. 

On the other hand, we only ask for correctness at the end of the stream, so we do not inherently require correctness at all times. 
Moreover, there exist special structured problems in the streaming model where correctness at all times does not require additional polylogarithmic overhead~\cite{BravermanCIW16,BlasiokDN17,BravermanCINWW17,Blasiok18}. 

\subsection{Our Contributions}
In this paper, we answer the question in the affirmative for insertion-only streams and in the negative for insertion-deletion streams. 
As a warm-up, we first give a simple algorithm for $(k,z)$-clustering on insertion-only streams that uses $o(\log n)$ words of space and nearly matches the same space complexity as the best offline coreset constructions. 

\begin{mdframed}[backgroundcolor=lightgray!40,topline=false,rightline=false,leftline=false,bottomline=false,innertopmargin=-4pt]
\begin{restatable}{theorem}{thmmain}
\thmlab{thm:main}
Given an accuracy parameter $\eps\in(0,1)$, an integer $k>0$ for the number of clusters, a constant $z\ge 1$, and a data stream consisting of points $X\subseteq[\Delta]^d$ with $X=\{x_1,\ldots,x_n\}$, there exists a one-pass streaming algorithm that uses $\tO{\frac{dk}{\varepsilon^2}}\cdot(2^{z\log z})\cdot\min\left(\frac{1}{\varepsilon^z},k\right)\cdot\poly(\log\log(n\Delta))$ words of memory and outputs a $(1+\eps)$-approximation to $(k,z)$-clustering at all times in the stream.    
\end{restatable}
\end{mdframed}

\thmref{thm:main} is the first result to achieve $(k,z)$-clustering on insertion-only streams using $o(\log n)$ words of space. 
That is, even for constant $d$, $k$, $\eps$, and $z$, existing works use $\poly(\log n)$ words of space, while \thmref{thm:main} achieves $\poly(\log\log n)$ words of space and in fact, its guarantees apply for all times in the stream. 
Moreover, we match or improve previous algorithms, e.g.,~\cite{Har-PeledM04,Har-PeledK07,Chen09,FeldmanL11,BravermanFLR19}, for insertion-only streams in terms of dependencies on $d$, $k$, and $\frac{1}{\eps}$ (see \figref{fig:summary} for a comprehensive summary). 
Furthermore, we match the best known \emph{offline} coreset constructions~\cite{Cohen-AddadLSS22} up to lower-order terms in the bit complexity. 



Despite its simplicity, we shall use \thmref{thm:main} as an opportunity to present several important structural properties that we will build upon. 

\paragraph{Online sensitivity sampling.}
First, we define the online sensitivities for $(k,z)$-clustering as follows:
\begin{definition}[Online sensitivity for $(k,z)$-clustering]
\deflab{def:online:sens}
Let $x_1,\ldots,x_n$ be a sequence of points with metric $\dist$ and let $X_t=x_1,\ldots,x_t$ for all $t\in[n]$. 
The \emph{online sensitivity} of $x_t$, denoted $\sigma_t$, is
\[\sigma_t:=\max_{C:|C|\le k}\frac{\Cost(x_t,C)}{\Cost(X_t,C)}=\max_{C:|C|\le k}\frac{\dist(x_t,C)^z}{\sum_{i=1}^t\dist(x_i,C)^z}.\]
\end{definition}
Sampling elements of the input with probability proportional to approximations of their sensitivities and online sensitivities, i.e., sensitivity sampling, and its variants have commonly been used in many other applications~\cite{DasguptaDHKM09,VaradarajanX12,CohenP15,CohenMM17,BravermanDMMUWZ20,HuangV20,Cohen-AddadSS21,MeyerMMWZ22,Cohen-AddadLSS22,MeyerMMWZ23,WoodruffY23}, because it is one of the most intuitive algorithms for acquiring a representative subset of the input -- each element of the input is sampled with probability proportional to a quantity that informally measures how ``important'' the element is. 

The sample complexity of sensitivity sampling is proportional to the total sensitivity, that is, the sum of the sensitivities of each element. 
Although in some cases analyzing the total sensitivity can be quite involved~\cite{DasguptaDHKM09,VaradarajanX12,CohenP15,CohenMM17,BravermanDMMUWZ20,WoodruffY23}, we give a simple proof upper bounding the sum of the online sensitivities:

\begin{mdframed}[backgroundcolor=lightgray!40,topline=false,rightline=false,leftline=false,bottomline=false,innertopmargin=-4pt]
\begin{restatable}{theorem}{thmtotalsens}(Upper bound on sum of online sensitivities)
\thmlab{thm:total:online:sens}
Let $X=\{x_1,\ldots,x_n\}\subset[\Delta]^d$ and for each point $x_t$ with $t\in[n]$, let $\sigma_t$ denote its online sensitivity for $(k,z)$-clustering for any $z\ge 1$. 
Then
\[\sum_{t=1}^n \sigma_t=\O{2^{2z}k\log^2(nd\Delta)}.\]
\end{restatable}
\end{mdframed}

We leverage \thmref{thm:total:online:sens} to analyze the space complexity of online sensitivity sampling in \thmref{thm:online:sens}, which provides a separate but sub-optimal approach for $(k,z)$-clustering on insertion-only streams. 
However, we utilize online sensitivity sampling as a crucial subroutine en route to achieving our main result in \thmref{thm:main}. 
We summarize these results in \figref{fig:summary}. 

\begin{figure}[!htb]
\centering
{
\tabulinesep=1.1mm
\begin{tabu}{|c|c|}\hline
Streaming algorithm & Words of Memory \\\hline\hline
\cite{Har-PeledK07}, $z\in\{1,2\}$ & $\tO{\frac{dk^{1+z}}{\eps^{\O{d}}}\log^{d+z} n}$ \\\hline
\cite{Har-PeledM04}, $z\in\{1,2\}$ &
$\tO{\frac{dk}{\eps^d}\log^{2d+2} n}$ \\\hline
\cite{Chen09}, $z\in\{1,2\}$ &
$\tO{\frac{d^2k^2}{\eps^2}\log^8 n}$ \\\hline
\cite{FeldmanL11}, $z\in\{1,2\}$ &
$\tO{\frac{d^2k}{\eps^{2z}}\log^{1+2z} n}$ \\\hline
Sensitivity and rejection sampling~\cite{BravermanFLR19} & $\tO{\frac{d^2k^2}{\eps^2}\log n}$ \\\hline
Online sensitivity sampling, i.e., \thmref{thm:online:sens} & 
$\tO{\frac{d^2k^2}{\eps^2}\log n}$ \\\hline
Merge-and-reduce with coreset of~\cite{Cohen-AddadLSS22} & $\tO{\frac{dk}{\eps^2}\log^4 n}\cdot\min\left(\frac{1}{\varepsilon^z},k\right)$ \\\hline\hline
This work, i.e., \thmref{thm:main} & $\tO{\frac{dk}{\eps^2}}\cdot\min\left(\frac{1}{\eps^z},k\right)\cdot\text{poly}(\log\log n)$ \\\hline
\end{tabu}
}
\caption{Table of $(k,z)$-clustering algorithms on insertion-only streams. We summarize existing results with $z=\O{1}$ and $\Delta=\poly(n)$ for the purpose of presentation.}
\figlab{fig:summary}
\end{figure}
\paragraph{Upper bounds, lower bounds and separations for dynamic streams.}
A natural follow-up question to ask is whether our results extend to dynamic streams, where points may be inserted and deleted in updates to the stream. 
We first show that a streaming algorithm that uses a single pass on a insertion-deletion stream cannot even provide a $2$-approximation to the \emph{cost} of an optimal $(k,z)$-clustering using $o(\log n)$ words of space. 
Note that by comparison, our algorithm in \thmref{thm:main} actually provides a strong coreset at the end of an insertion-only stream and thus not only outputs a set of $k$ near-optimal centers, but also a $(1+\O{\eps})$-approximation to the cost induced by those $k$ centers and by extension, a $(1+\eps)$-approximation to the optimal $(k,z)$-clustering cost. 
Since the following statement, \thmref{thm:dynamic:lb}, shows that any algorithm that provides a $2$-approximation on dynamic streams requires $\Omega(\log^2 n)$ bits, i.e., $\Omega(\log n)$ words of space, then we cannot expect to achieve such a result for dynamic streams.

\begin{mdframed}[backgroundcolor=lightgray!40,topline=false,rightline=false,leftline=false,bottomline=false,innertopmargin=-4pt]
\begin{restatable}{theorem}{thmdynamiclb}
\thmlab{thm:dynamic:lb}
Let $z$ be a constant. 
Then even for $k=1$, any algorithm that with probability at least $\frac{2}{3}$, simultaneously outputs a $2$-approximation to the optimal $(k,z)$-clustering cost at all times of a dynamic stream of length $n$ for points in $\{0,1\}^d$ with $d=\Omega(\log n)$ must use $\Omega(\log^2 n)$ bits of space.  
\end{restatable}
\end{mdframed} 

At first glance, \thmref{thm:dynamic:lb} may not seem like a separation from \thmref{thm:main} because the former requires the assumption that $d=\Omega(\log n)$ and the latter has a linear dependence on $d$. 
However, we show that if the goal is just to output an estimation to the \emph{cost} of the optimal $(k,z)$-clustering at the end of an insertion-only stream, then our upper bounds do not require a linear dependence on $d$: 

\begin{mdframed}[backgroundcolor=lightgray!40,topline=false,rightline=false,leftline=false,bottomline=false,innertopmargin=-4pt]
\begin{restatable}{theorem}{thmclustercost}
\thmlab{thm:cluster:cost}
Given an accuracy parameter $\eps\in(0,1)$, an integer $k>0$ for the number of clusters, a constant $z>0$, and a data stream consisting of points $X\subseteq[\Delta]^d$ with $X=\{x_1,\ldots,x_n\}$, there exists a one-pass streaming algorithm that uses $\tO{\frac{k}{\eps^2}}\min\left(\frac{1}{\eps^z},k\right)\cdot\poly(\log\log(n\Delta))$ words of memory and outputs a $(1+\eps)$-approximation to the cost of the optimal $(k,z)$-clustering at all times in the stream. 
\end{restatable}
\end{mdframed}
Combined with \thmref{thm:dynamic:lb}, \thmref{thm:cluster:cost} provides a separation for the $(k,z)$-clustering problem between insertion-only and dynamic streams. 

We also show that any one-pass dynamic streaming algorithm that computes a constant factor-approximation to the clustering cost from a weighted sample of the input points must use $\Omega(\log^2 n)$ bits of space.
\begin{mdframed}[backgroundcolor=lightgray!40,topline=false,rightline=false,leftline=false,bottomline=false,innertopmargin=-4pt]
\begin{restatable}{theorem}{thmonepassdynlb}
\thmlab{thm:one:pass:dyn:lb}
Any one-pass dynamic streaming algorithm that computes a $2$-approximation to the clustering cost from a weighted sample of the input points with probability at least $\frac{2}{3}$ must use $\Omega(\log^2 n)$ bits of space. 
\end{restatable}
\end{mdframed}

On the other hand, we show that the $\Omega(\log^2 n)$ space barrier can be broken on dynamic streams if we permit algorithms an additional pass over the stream. 
\begin{mdframed}[backgroundcolor=lightgray!40,topline=false,rightline=false,leftline=false,bottomline=false,innertopmargin=-4pt]
\begin{theorem}
\thmlab{thm:dynamic:two:main}
There exists a two-pass dynamic streaming algorithm that outputs a $(1+\eps)$-coreset for $k$-median and $k$-means clustering, with probability at least $\frac{2}{3}$, and uses $\tO{\frac{1}{\eps^2}}\cdot\poly(d,k,\log\log(n\Delta))$ words of space.
\end{theorem}
\end{mdframed}
We remark that \thmref{thm:dynamic:two:main} can be generalized to $(k,z)$-clustering for all $z\in[1,2]$.

\paragraph{Wasserstein-$z$ bicriteria embedding.}
Along the way, we introduce a quadtree embedding technique that achieves a structured bicriteria approximation to $(k,z)$-clustering and is key to our two-pass dynamic streaming algorithm. 
Although quadtree embeddings are popular for $k$-median clustering in big data models, generalizations to $(k,z)$-clustering are not known for a single quadtree and in fact, there are simple examples on a line that show that the expected distortion between the squared distances of a set of $n$ points, i.e., $z=2$, and the estimated distance by a quadtree is $\Omega(n)$, e.g., see \appref{app:quadtree:means:bad}. 
Recently, \cite{Cohen-AddadLNSS20} overcame this barrier in the offline setting by considering multiple quadtrees and taking the minimum estimated distance for pairs of points across the quadtrees, but this approach does not seem to work for the streaming setting, because we will no longer be able to estimate these distances after embedding the quadtree into $L_1$. 

We instead develop a quadtree embedding technique for $(k,z)$-clustering that uses a bicriteria approximation, which often suffices for downstream applications. 
Indeed, we crucially use our embedding as a subroutine toward our $(1+\eps)$-approximate streaming algorithm for $(k,z)$-clustering in \thmref{thm:main}. 
Our embedding has implications to the Wasserstein-$z$ distance, which is a distance between probability measures and corresponds to the important earth mover distance for $z=1$. 
We use $\WASSZZ$ to denote the $z$-th power of the Wasserstein-$z$ distance and formally recall its definition in \secref{sec:prelims}. 
Our embedding then has the following guarantees:

\begin{mdframed}[backgroundcolor=lightgray!40,topline=false,rightline=false,leftline=false,bottomline=false,innertopmargin=-4pt]
\begin{restatable}{theorem}{thmquadtreebicrit}
\thmlab{thm:quadtree:bicrit}
Let $\mu,\nu\in\mathbb{R}^{[\Delta]^d}$ be probability measures such that $\nu$ has support at most $k$ on a set $C\subset[\Delta]^d$.  
Then there exists a quadtree embedding $W_s$ parameterized by a random shift parameter $s$ and explicit mappings $\psi,\phi$ such that with probability at least $0.99$, $\psi(\nu)$ is a probability mass with support at most $\O{k}$ on the set $\phi(C)$ and
\[\|W_s(\mu-\psi(\nu))\|_1\le\O{d^{1+0.5z}\log^{z-1}\Delta}\cdot\WASSZZ(\mu,\nu).\]
\end{restatable}
\end{mdframed}

In fact, the set $\phi(C)$ includes the set $C$, so $\psi(\nu)$ is a probability measure that includes the support of $\nu$. 
Thus, our embedding can be viewed as a sketch that produces a bicriteria approximation for the Wasserstein-$z$ distance. 

\subsection{Technical Overview: Initial Challenges}
In this section, we describe the intuition behind our main algorithm and how it overcomes significant barriers for previous techniques. 

\paragraph{Merge-and-reduce does not work.}
A standard approach for $(k,z)$-clustering on datasets of $\mathbb{R}^d$ on insertion-only streams is the merge-and-reduce framework, due to \cite{BentleyS80,Har-PeledM04}.  
Given a coreset construction algorithm for $(k,z)$-clustering with size $S(n,d,k,\eps,\delta)$ where $\eps$ is the desired accuracy and $\delta$ is the failure probability, the merge-and-reduce framework first partitions the stream into consecutive blocks of size $S(n,d,k,\eps',\delta')$, where $\eps'=\frac{\eps}{\O{\log n}}$ and $\delta'=\frac{\delta}{\poly(n)}$. 
A coreset with accuracy $(1+\eps')$ and failure probability $\delta'$ is then computed for each block, so that each coreset also uses space $S(n,d,k,\eps',\delta')$. 
We can then view these coresets as the leaves of a binary tree of height $\O{\log n}$, where each node in the tree at depth $t$ denotes a coreset computed from two coresets at depth $t+1$. 
See \figref{fig:merge:reduce} for an example of the merge-and-reduce framework. 

The coreset at the root node represents a coreset for the entire stream, and since each node only requires the coresets of its children nodes, this process can be done on-the-fly. 
Because there is a $(1+\eps')$ multiplicative loss of accuracy at each level, then the coreset at the root node has accuracy $(1+\eps')^{\O{\log n}}=(1+\eps)$, as desired. 
However, all known coreset constructions use $\Omega\left(\frac{1}{\eps'}\right)$ space, and thus for $\eps'=\frac{\eps}{\O{\log n}}$, the merge-and-reduce approach would not work for our goal because it incurs extraneous $\log n$ factors. 

\begin{figure*}[tb]
\centering
\begin{tikzpicture}[scale=0.4]

\node at (-2.3,0.7){Stream:};
\draw (0,0.25) rectangle+(8,0.5);
\draw (8,0.25) rectangle+(8,0.5);
\draw (16,0.25) rectangle+(8,0.5);
\draw (24,0.25) rectangle+(8,0.5);

\draw (4,0.8) -- (4,1.4);
\draw (12,0.8) -- (12,1.4);
\draw (20,0.8) -- (20,1.4);
\draw (28,0.8) -- (28,1.4);

\node at (-2.3,0.7+1.3){Depth 3:};
\draw (0+0.1,1.5) rectangle+(7.8,1);
\draw (8+0.1,1.5) rectangle+(7.8,1);
\draw (16+0.1,1.5) rectangle+(7.8,1);
\draw (24+0.1,1.5) rectangle+(7.8,1);
\node at (4,2){\tiny{$C_{3,1}$}};
\node at (12,2){\tiny{$C_{3,2}$}};
\node at (20,2){\tiny{$C_{3,3}$}};
\node at (28,2){\tiny{$C_{3,4}$}};

\draw[dashed] (4,1.1+1.5) -- (7,1.4+1.5);
\draw[dashed] (12,1.1+1.5) -- (9,1.4+1.5);
\draw[dashed] (20,1.1+1.5) -- (23,1.4+1.5);
\draw[dashed] (28,1.1+1.5) -- (25,1.4+1.5);

\node at (-2.3,0.7+1.3+1.5){Depth 2:};
\draw (4+0.1,3) rectangle+(7.8,1);
\draw (20+0.1,3) rectangle+(7.8,1);
\node at (8,3.5){\tiny{$C_{2,1}$}};
\node at (24,3.5){\tiny{$C_{2,2}$}};

\draw[dashed] (8,1.1+1.5*2) -- (15,1.4+1.5*2);
\draw[dashed] (24,1.1+1.5*2) -- (17,1.4+1.5*2);

\node at (-2.3,0.7+1.3+1.5*2){Depth 1:};
\draw (12+0.1,4.5) rectangle+(7.8,1);
\node at (16,5){\tiny{$C_{1,1}$}};

\end{tikzpicture}
\caption{Merge and reduce framework. Each coreset $C_{t,i}$ at depth $t$ is a coreset of the merger of the two children coresets $C_{t+1,2i-1}$ and $C_{t+1,2i}$ at depth $t+1$.}
\figlab{fig:merge:reduce}
\end{figure*}
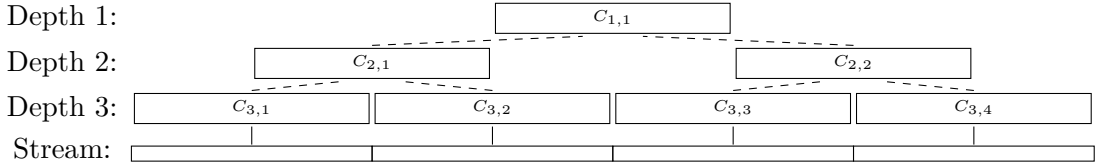

\paragraph{Quadtrees are not enough.}
Another approach for $(k,z)$-clustering on dynamic streams is placing the input space into a quadtree~\cite{Chen09}, which is a tree of randomly shifted grids on $[\Delta]^d$. 
In the quadtree, the coarsest grid contains a single cell representing the entire input space of $[\Delta]^d$ and each subsequent level of the quadtree corresponds to a refinement of the previous grid with smaller grid cells, e.g., by splitting each grid cell in the previous level into $2^d$ smaller grid cells. 
It then suffices to identify the finest level that contains a number of non-empty cells below a certain threshold, as well as the non-empty cells and the number of points in those non-empty cells at that level, e.g., using sparse recovery~\cite{BravermanFLSY17,SongYZ18}. 
However, because the number of cells in the quadtree is $2^{\O{d}}$, then the frequency vector of the number of points in each cell has dimension $2^{\O{d}}$ and thus recovery of the non-empty cells and the number of points in those cells often requires $\O{d}$ space, which, combined with subsequent union bounds over nets with size $\left(\frac{n\Delta}{\eps}\right)^{\O{kd}}$, results in $\poly(d,\log(n\Delta))$ dependencies prohibitive for our goal and in fact even the $\poly(d)$ dependencies cannot be removed through known dimensionality reduction techniques~\cite{MakarychevMR19,IzzoSZ21,CharikarW22,CharikarW22b} without incurring extraneous $\poly\left(\frac{1}{\eps}\right)$ factors. 

\paragraph{Offline sensitivity sampling is not enough.}
Another approach would be to adapt recent methods for offline coreset constructions to the streaming model. 
Unfortunately, many of the optimal or near-optimal constructions~\cite{HuangV20,Cohen-AddadSS21,Cohen-AddadLSS22} use multiple stages of sampling, which seems inherently infeasible for one-pass algorithms in the streaming model. 

Luckily, the sensitivity framework~\cite{FeldmanL11,FeldmanS12,BravermanFLSZ21} for coreset construction seems amenable to adaptation to data streams. 
The (offline) \emph{sensitivity} of each point $x_1,\ldots,x_n\in\mathbb{R}^d$ of a set $X$ is a quantity that informally measures the ``importance'' of that point. 
Formally, the sensitivity of $x_t$ for the $(k,z)$-clustering problem is defined as
\[\max_{C\subset\mathbb{R}^d: |C|\le k}\frac{\Cost(x_t,C)}{\Cost(X,C)}=\max_{C\subset\mathbb{R}^d: |C|\le k}\frac{\dist(x_t,C)^z}{\sum_{t=1}^n\dist(x_t,C)^z}.\]
Traditionally, the sensitivity sampling framework independently samples a fixed number of points with replacement, so that each point is sampled with probability proportional to (some approximation to) its sensitivity. 
However, it can be shown that sampling each point independently without replacement, so that the total number of sampled points is a random variable, is also a valid coreset construction algorithm. 
Recent works have considered online variants of sensitivity sampling for other problems such as subspace embedding and low-rank approximation~\cite{CohenMP20,BravermanDMMUWZ20,WoodruffY23}. 

\paragraph{Online sensitivity sampling.}
We thus consider online sensitivity sampling for $(k,z)$-clustering. 
The argument of correctness is relatively straightforward. 
We first fix a set $C$ of $k$ centers and show that the expectation of the cost of clustering with $C$ for the sampled points is an unbiased estimator of the clustering cost of $C$ with respect to the input set $X$. 
We then upper bound the variance of the cost of clustering the coreset with $C$, which allows us to apply a standard martingale argument that shows concentration, i.e., the coreset approximately preserves the clustering cost with respect to $C$.   
Although there is an arbitrary number of subsets of $\mathbb{R}^d$ of size $k$, it is well-known that to achieve an approximately optimal clustering, it suffices to only show correctness on a net of size $\left(\frac{n}{\eps}\right)^{\O{kd}}$, and to adjust the probability of failure and apply a union bound. 

Since each point is sampled with probability proportional to its online sensitivity, the total number of points sampled is proportional to the sum of the online sensitivities of the points. 
We thus upper bound the total online sensitivity as follows. 
We first break the stream into $\O{\log n}$ blocks where the cost of the optimal $(k,z)$-clustering doubles. 
For each block, we show that the online sensitivity of the points that arrive in a block that ends at time $t$ can be charged to either the cost of an optimal clustering $K_t$ for the stream up to $t$ or to the number of points in some cluster of $K_t$. 

More specifically, because the online sensitivity of a point $x_t$ is defined as
\[\max_{C:\,C\subset\mathbb{R}^d,\, |C|=k}\frac{\dist^z(x_t,C)}{\sum_{i=1}^t\dist^z(x_i,C)},\]
then we can use the generalized triangle inequality to upper bound the online sensitivity of $x_t$ by 
\[\frac{\dist^z(x_t,C)}{\sum_{i=1}^t\dist^z(x_i,C)}\le \frac{2^{z-1}\dist^z(x_t,\pi(x_t))}{\sum_{i=1}^t\dist^z(x_i,C)}+\frac{2^{z-1}\dist^z(\pi(x_t),C)}{\sum_{i=1}^t\dist^z(x_i,C)},\]
where $\pi$ is the mapping to the closest point in $K_t$. 
We can further lower bound the denominator $\sum_{i=1}^t\dist^z(x_i,C)$ by $2\sum_{i=1}^t\dist^z(x_i,\pi(x_i))$ using the optimality of $K_t$, therefore upper bounding the sum of the first term across all $t$. 
We can also use a charging argument to upper bound the second term by $\O{\frac{1}{|S_t|}}$, where $S_t$ is the subset of $\{x_1,\ldots,x_t\}$ mapped to $\pi(x_t)$ at time $t$. 
Since there are $k$ possible clusters for $\pi(x_t)$, $\sum_{i=1}^n\frac{1}{i}=\O{\log n}$, and the stream is partitioned into $\O{\log n}$ blocks, this informally gives our $\O{k\log^2 n}$ upper bound for the total online sensitivity, i.e., \thmref{thm:total:online:sens}. 

\subsection{Technical Overview: Algorithmic Intuition for Insertion-Only Streams}
Despite these efforts, online sensitivity sampling is insufficient for our ultimate goal because it has an  $\O{dk^2\log^2(n\Delta)}$ dependency, which is prohibitive in both the $k$ and $\log(nd\Delta)$ factors. 
However, a crucial observation is that the points sampled by online sensitivity sampling form a stream $\calS'$ of length $\poly\left(k,d,\log(nd\Delta),\frac{1}{\eps}\right)$ of weighted points that forms a coreset of the input points. 
Thus if we could somehow access $\calS'$ in a single pass without storing all the points of $\calS'$, then the length of $\calS'$ is now small enough for us to  simply run a merge-and-reduce algorithm on $\calS'$. 

Unfortunately, online sensitivity sampling requires the storage of the entire set of sampled points to compute the online sensitivities of future points. 
A natural idea would be to use a data structure that could give ``good'' approximations to the online sensitivities without using prohibitively large space. 
However, if we require correctness at all times, it seems likely that any analysis that essentially requires a union bound over all times $t\in[n]$ would result in prohibitive $\log(n\Delta)$ factors. 
Thus, we instead settle for a data structure that could give ``good'' approximations to the online sensitivities ``most'' of the time, without using prohibitively large space. 

Our main insight is that any coreset to the underlying dataset at time $t-1$ precisely gives this guarantee at time $t$. 
That is, suppose we have a $(1+\eps)$-coreset at time $t-1$. 
We can use the coreset to compute approximations to the online sensitivity of $x_t$, which we can then use to sample $x_t$ into a conceptual stream $\calS'$. 
Because $\calS'$ is formed by online sensitivity sampling, then at all times $t\in[n]$, the optimal clustering to the set of weighted points in $\calS'$ at time $t$ is a $(1+\eps)$-approximation to the optimal clustering for the first $t$ points of the original stream. 
We can then apply a standard merge-and-reduce algorithm on $\calS'$ to obtain a $(1+\eps)$-coreset to $\calS'$, which translates to a $(1+\O{\eps})$-coreset for the original stream and guarantees that we now have a coreset for time $t$. 
We can then iterate on this process to perform online sensitivity sampling at time $t+1$. 

Our analysis for online sensitivity sampling shows that with high probability, the number of sampled points into $\calS'$ will be $\poly\left(k,d,\log(nd\Delta),\frac{1}{\eps}\right)$. 
Crucially, since merge-and-reduce requires space that is polylogarithmic in the length of the input stream, then running merge-and-reduce on $\calS'$ will use $o(\log n)$ words of space. 
We summarize our approach in \figref{fig:flowchart}. 

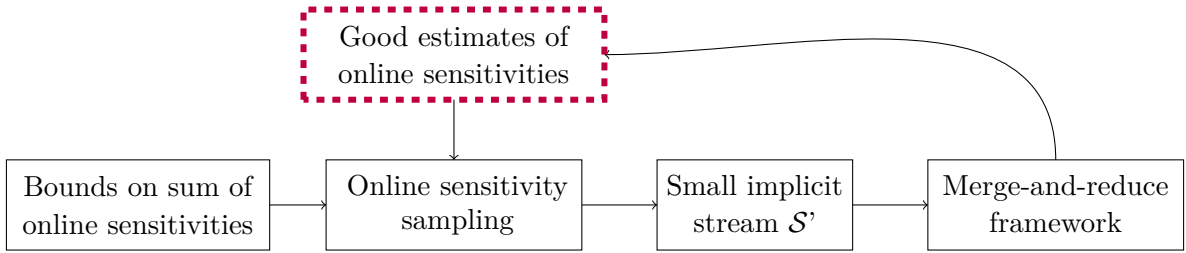
\begin{figure*}[!htb]
\centering
\begin{tikzpicture}[scale=1]
\draw (-0.25,0) rectangle+(3.5,1.2); 
\node at (3/2,0.85){Bounds on sum of};
\node at (3/2,0.35){online sensitivities};
\draw[->] (3.25,1.2/2) -- (4,1.2/2);

\draw (4,0) rectangle+(3.4,1.2); 
\node at (4+3.5/2,0.83){Online sensitivity};
\node at (4+3.5/2,0.37){sampling};
\draw[->] (7.4,1.2/2) -- (8.4,1.2/2);

\draw (8.4,0) rectangle+(2.6,1.2); 
\node at (8.4+2.6/2,0.83){Small implicit};
\node at (8.4+2.6/2,0.37){stream $\calS$'};
\draw[->] (11,1.2/2) -- (12,1.2/2);

\draw (12,0) rectangle+(3.4,1.2); 
\node at (12+3.4/2,0.83){Merge-and-reduce};
\node at (12+3.4/2,0.37){framework};
\draw[->] (13.7,1.2) to[out=90,in=0] (7.7,2.6);

\draw[dashed,line width=0.8mm,color=purple] (3.7,0+2) rectangle+(4,1.2); 
\node at (4+4/2-0.3,0.85+2){Good estimates of};
\node at (4+4/2-0.3,0.35+2){online sensitivities};
\draw[->] (5.7,2) -- (5.7,1.2);
\end{tikzpicture}
\caption{Flowchart of our simple algorithm and analysis for $(k,z)$-clustering on insertion-only streams. .}
\figlab{fig:flowchart}
\end{figure*}

\subsection{Technical Overview: Algorithmic Intuition for Dynamic Streams}
Unfortunately, merge-and-reduce approaches generally do not seem to work for insertion-deletion, i.e., dynamic streams. 
Similarly, online sensitivity sampling and its variants do not seem to immediately work for dynamic streams. 
Consider the following definition for the sensitivity of a point $x\in[\Delta]^d$ with respect to a dataset $X=\{x_1,\ldots,x_n\}$, defined as
\[\max_{C:\,C\subset\mathbb{R}^d,\, |C|=k}\frac{\dist^z(x,C)}{\sum_{i=1}^n\dist^z(x_i,C)}.\]
The issue for dynamic streams is twofold, both related to the possible subsequent deletion of points: 1) sampled points can later be removed from the stream, 2) the online sensitivity of a point $x$ at some time in the stream may be less than the sensitivity of the point $x$ with respect to $X$, because the data stream may later remove additional points. 
That is, the online sensitivity of a point at the time it arrives in the data stream can be either significantly higher or significantly lower than its sensitivity at the end of the stream. 
For example, let $X$ be the dataset at the end of the stream and let $Y$ be an additional set of points so that $X\cup Y$ is the set of points when $x$ arrives in the stream. 
Then the sensitivity of $x$ is defined with respect to $X$, while the online sensitivity of $x$ is by definition, the sensitivity of $x$ with respect to $X\cup Y$, and thus less than the true sensitivity of $x$ with respect to $X$. 
Although we do not see how to overcome these issues in a single pass over the data stream and indeed we give impossibility results for $o(\log n)$ space (in words of memory) algorithms for one-pass dynamic streams, fortunately it seems these issues can be resolved using a second pass. 

We first show that an approximately optimal clustering can be used to compute approximations of the sensitivities of the points.
Therefore, our hope is to use a data structure that could give ``good'' approximations to the optimal clustering, without using prohibitively large space. 
To that end, our algorithm uses two main steps: 1) we first embed the Wasserstein-$z$ distance into $L_1$ and then 2) we use a streaming algorithm to estimate the $L_1$ distance. 
For the ease of discussion, we first consider $z=1$, i.e., $k$-median clustering, so that the Wasserstein-$z$ distance corresponds to the Earth Mover Distance (EMD). 

\paragraph{EMD sketch for approximate $k$-median clustering cost.}
To embed $\EMD$ on $[\Delta]^d$ into $L_1$, we first generalize a standard quadtree consisting of refinements of a randomly shifted grid, i.e.,~\cite{indyk120oct,BackursIRW16}. 
Suppose without loss of generality that $\Delta=2^\ell$ for some non-negative integer $\ell$. 
Given a shift parameter $s=(s_1,\ldots,s_d)\in\mathbb{Z}^d$ and $t\in\{0,1,\ldots,\ell+1\}$, we define the linear map $G_{s,t}$ from $\mathbb{R}^{[\Delta]^d}$ into $L_1$ by first defining the grid $\calG_{s,t}$ over $\mathbb{Z}^d$ with side length $2^t$, so that $s=(s_1,\ldots,s_d)$ lies on one of the corners of the grid. 
For $\mu\in\mathbb{R}^{[\Delta]^d}$, we define $G_{s,t}\mu$ as the frequency vector whose coordinates correspond to the total mass in each cell/hypercube of the grid $\calG_{s,t}$. 
Then we define the mapping $G_s\mu$ to be the concatenation of the vectors 
\[G_s\mu=(G_{s,0}\mu)\circ(2\cdot G_{s,1}\mu)\circ\ldots\circ(2^t\cdot G_{s,t}\mu)\circ\ldots\circ(2^\ell\cdot G_{s,\ell}\mu).\]

The main intuition is that for a mass vector $\mu$ and an assignment vector $\nu$ on $[\Delta]^d$, the cost to the optimal transport induced by some mass in $\mu$ and $\nu$ corresponds to the finest grid the mass does not appear in the same cell. 
Due to the scaling of each grid, the $L_1$ mass attributed to the frequency vector $G_s(\mu-\nu)$ is proportional to the cost needed to move the mass from $\mu$ to $\nu$. 
Thus the grid embedding not only gives a good approximation to EMD, but also enjoys the property that it can only overestimate the EMD, which we shall utilize in the analysis. 

We also remark that compared to other grid embeddings, our data structure achieves worse approximation guarantees but crucially uses less space. 
In fact, the worse approximation guarantees ultimately only means that we sample more points into the implicit stream $\calS'$ in such a way that the downstream guarantees of running a merge-and-reduce algorithm on $\calS'$ will not be affected. 
Given the embedding of $\EMD$ into $L_1$, we then use a streaming algorithm to estimate the $L_1$ distance. 
The guarantee of the $L_1$ streaming algorithm ensures that after the first pass, the accuracy of the estimated costs of the optimal clusterings ensures that we have a good approximation to the sensitivity of each point. 


\paragraph{Putting things together for $k$-median.}
Given two passes over the stream, we can use the first pass over the data stream to maintain the sketch of the EMD embedding. 
Importantly, the EMD embedding and the corresponding $L_1$ sketch are both linear sketches, and so they can handle both insertions and deletions.

We then use the second pass over the data stream to perform sensitivity sampling, since we can now compute the sensitivity of each point with respect to the final data set $X$ at the end of the stream. 
Unfortunately, sensitivity sampling is still  incompatible with the dynamic setting, since sampled points could be subsequently deleted. 
Thus, we instead subsample from the universe with probability proportional to each point's sensitivity, so that if a universe element is subsampled, all updates to the element throughout the stream are reported. 
We then use a sparse recovery scheme to track the updates to the sampled universe elements. 
Since there are known linear sketches for sparse recovery, then our algorithm can handle both insertions and deletions while simulating sensitivity sampling. 
Thus at the end of the stream, our algorithm outputs a weighted subset of the input points that forms a $(1+\eps)$-coreset of the underlying point set for $k$-median clustering. 

\subsection{Technical Overview: Embedding for \texorpdfstring{$(k,z)$}{(k,z)}-Clustering}
To extend our approach to $(k,z)$-clustering, it seems crucial to develop an analog for EMD sketching for the Wasserstein-$z$ distance. 
Unfortunately despite substantial efforts, such an embedding is not known for a single quadtree, because the distortion incurred by a splitting hyperplane in the quadtree is exponential in $z$, while the probability that two points are split by a hyperplane is inversely linear in their distance. 
In fact, there are simple examples on a line that show that the expected distortion between the squared distances of a set of $n$ points, i.e., $z=2$, and the estimated distance by a quadtree is $\Omega(n)$. 
See \appref{app:quadtree:means:bad} for one such example. 
Hence while $k$-median, i.e., $z=1$ is ideal for the quadtree approach, larger values of $z$ can have unbounded distortion. 

Recently, \cite{Cohen-AddadLNSS20} overcame this barrier in the offline setting by considering multiple quadtrees and taking the minimum estimated distance for pairs of points across the quadtrees. 
This approach does not seem to work for the streaming setting, because we can only separately embed each of the quadtrees into $L_1$ and then we will only have access to the estimated clustering costs by each of the quadtrees. 
However, the minimum estimate of the sum of the squared distances is not equal to the sum of the minimum estimated squared distances; the former is what the embedding into $L_1$ would give, but the latter is what the approach of \cite{Cohen-AddadLNSS20} requires. 

\paragraph{Wasserstein-$z$ embedding preliminaries.}
We instead develop a quadtree embedding technique that achieves a bicriteria approximation for $(k,z)$-clustering. 
As in the EMD sketch, we select a shift $s=(s_1,\ldots,s_d)\in[\Delta]^d$ uniformly at random and then add $s$ to each of the input points $x_1,\ldots,x_n\in[\Delta]^d$. 
We create a quadtree so that the root of the tree represents a grid with side length $2\Delta$. 
We then partition the grid into $2^d$ smaller hypercubes with side length $\Delta$, so that for each of these hypercubes that contains a shifted input point, we create a node representing the hypercube and add the node as a child of the root note in the tree, using an edge with weight $\sqrt{d}\Delta$. 
We repeat this procedure until every node representing a hypercube contains at most a single point, resulting in a tree where all leaves contain a single point and are all at the same height, which is at most $\O{\log\Delta}$. 

We embed the input points $X$ into the quadtree. 
As it will be useful to index from the leaves of the tree, we define the $t$-th level of the quadtree to correspond with the hypergrid with cells of length $2^t$. 
For $\mu\in\mathbb{R}^{[\Delta]^d}$, we define $W_{s,t}\mu$ to be the frequency vector over the hypercubes of the hypergrid $\calG_{s,t}$ at level $t$ that counts the total mass in each hypercube. 
We define the mapping $W_s\mu$ to be the concatenation of the vectors 
\[W_s\mu=(W_{s,0}\mu)\circ((2\sqrt{d})^z\cdot W_{s,1}\mu)\circ\ldots\circ((2^z\sqrt{d})^{z}\cdot W_{s,t}\mu)\circ\ldots\circ((2^\ell\sqrt{d})^{z}\cdot W_{s,\ell}\mu).\]
We remark that up to this point, the approach is the same as previous quadtree embeddings. 

\paragraph{Wasserstein-$z$ embedding through bicriteria approximation.}
Now, for a query point $q$, we say that $q$ is \emph{bad at level $i$} if there exists a hyperplane of the quadtree decomposition of length $2^i$ that has distance less than $\frac{2^i}{d\log\Delta}$ from $q$. 
Otherwise, we say that $q$ is good at level $i$. 
Observe that if $q$ is good at level $i$, but $q$ and $x\in X$ are separated at level $i$ but not $i+1$, then $\|q-x\|_2^z>\frac{2^{iz}}{d^{z}\log^{z}\Delta}$, since all hyperplanes at level $i$ are at least distance $\frac{2^i}{d^2\log^2\Delta}$ from $q$. 

Moreover since $q$ and $x$ are not separated at level $i+1$, then the incurred estimated cost for $q$ and $x$ in the quadtree is most $(2^i\sqrt{d})^z$, so the distortion will be at most $d^{0.5z}\log^{z}\Delta$. 

On the other hand, if $q$ is bad at level $i$, then by definition, it has distance less than $\frac{2^i}{d\log\Delta}$ from a hyperplane of the quadtree decomposition of length $2^i$ and so $q$ and $x$ incur estimated cost $2^{iz}d$ by the quadtree, then the distortion could be significantly larger. 
To address this issue, we define a mapping $\phi$ to subsets $S_q$ of $\mathbb{R}^d$ for a query point $q$ as follows. 
We first add $q$ to $S_q$. 
If $q$ is bad at level $i$, then consider each hyperplane $H$ of the quadtree of length $2^i$ that is too close to $q$, i.e., $\dist(q,H)\le\frac{2^i}{d\log\Delta}$. 
We create a copy $q^{(H)}_i$ corresponding to the projection of $q$ onto $H$, so that $q^{(H)}_i$ is responsible for serving the points assigned to $q$ that are on the other side of $H$. 
We add $q^{(H)}_i$ to $S_q$ and proceed top-down, repeatedly adding points to $S_q$ as necessary. 
See \figref{fig:wassz:quadtree} for an example of this process. 

Note that if $q$ is bad at level $i$, it could be too close to multiple hyperplanes of the quadtree with length $2^i$; in this case, we add a point $q^{(H)}_i$ to $S_q$ for each hyperplane $H$ for which $\dist(q,H)\le\frac{2^i}{d^2\log^2\Delta}$. 
That is, $q$ could induce multiple points $q^{(H)}_i$ to be added to $S_q$. 

\setlength{\columnsep}{1cm}
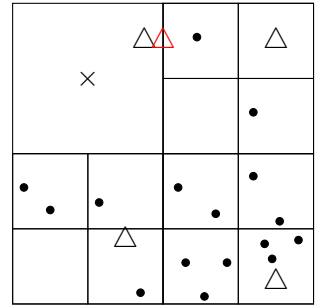
\begin{wrapfigure}{r}{5.5cm}
\centering
\begin{tikzpicture}[scale=0.5]
\draw (0,0) rectangle+(4,4); 
\draw (4,0) rectangle+(4,4); 
\draw (0,4) rectangle+(4,4); 
\draw (4,4) rectangle+(4,4);

\draw (0,0) rectangle+(2,2); 
\draw (2,0) rectangle+(2,2); 
\draw (0,2) rectangle+(2,2); 
\draw (2,2) rectangle+(2,2);

\draw (4,4) rectangle+(2,2); 
\draw (6,4) rectangle+(2,2); 
\draw (4,6) rectangle+(2,2); 
\draw (6,6) rectangle+(2,2);

\draw (4,0) rectangle+(2,2); 
\draw (6,0) rectangle+(2,2); 
\draw (4,2) rectangle+(2,2); 
\draw (6,2) rectangle+(2,2);

\filldraw (0.3,3.1) circle (0.1);
\filldraw (1,2.5) circle (0.1);
\filldraw (2.3,2.7) circle (0.1);
\filldraw (3.4,0.3) circle (0.1);
\node at (3,1.7){$\triangle$};

\node at (2,6){$\times$};
\node at (3.5,7){$\triangle$};
\node at (4,7){\color{red}{$\triangle$}};

\filldraw (4.9,7.1) circle (0.1);
\filldraw (6.4,5.1) circle (0.1);
\node at (7,7){$\triangle$};

\filldraw (4.4,3.1) circle (0.1);
\filldraw (4.6,1.1) circle (0.1);
\filldraw (5.1,0.2) circle (0.1);
\filldraw (5.4,2.4) circle (0.1);
\filldraw (5.7,1.1) circle (0.1);
\filldraw (6.4,3.4) circle (0.1);
\filldraw (6.7,1.6) circle (0.1);
\filldraw (6.9,1.2) circle (0.1);
\filldraw (7.1,2.2) circle (0.1);
\filldraw (7.6,1.7) circle (0.1);
\node at (7,0.6){$\triangle$};
\end{tikzpicture}
\caption{Example of quadtree for Wasserstein-$z$ embedding. 
Black circles are the input points and black triangles denote the query centers. 
One of the four queries (top left cell) is too close to splitting hyperplanes, which could cause too much distortion, since the query is mapped to the cell centers denoted by the black X. 
Hence, we create an additional center on the cell boundary, marked by the red triangle.}
\figlab{fig:wassz:quadtree}
\vspace{-0.5in}
\end{wrapfigure}

By construction, the resulting mapping $\phi(C):=\cup_{j=1}^k\phi(C_j)$ for a set of $k$ centers $C=\{C_1,\ldots,C_k\}$ will induce an estimated cost by the quadtree that has small distortion, i.e., $d^{0.5z}\log^{z}\Delta$, from the actual cost. 
With more careful fine-grained analysis, we show that the expected distortion is at most $\O{d^{1+0.5z}\log^{z-1}\Delta}$. 
To ensure the bicriteria guarantee, it remains to show that $\phi(C)$ contains at most $\O{k}$ centers in expectation. 

Consider a level $i$ and a center $C_j$ with $j\in[k]$. 
Let $\calH$ be the set of hyperplanes of the hypergrid at level $i$. 
By a union bound over the $d$ dimensions, the probability that $\dist(C_j,\calH)<\frac{2^i}{d\log\Delta}$ is at most $\O{\frac{d}{d\log\Delta}}=\O{\frac{1}{\log\Delta}}$. 
Hence, summing up over all the $\O{\log\Delta}$ levels, the expected number of center copies that a center $C_j$ can generate into $\phi(C)$ is at most $\O{1}$. 
Finally, summing up over $j\in[k]$, the expected number of generated points is at most $\O{k}$ and thus the expected size of $\phi(C)$ is at most $k+\O{k}=\O{k}$. 

\paragraph{Putting things together for $(k,z)$-clustering.}
To adapt our two-pass dynamic streaming $k$-median clustering algorithm for $k$-means clustering, it is instructive to consider which steps break down. 
Recall that the first pass of our $k$-median algorithm was used to set up the EMD embedding, which subsequently gave a $\O{\log k+\log\log\Delta}$ approximation to the sensitivity of each point in the second pass, so that we could perform sensitivity sampling through sparse recovery in the second pass. 

For $k$-means clustering, the natural approach would be to replace the EMD embedding with our Wasserstein-$z$ embedding. 
However, our Wasserstein-$z$ embedding may distort the estimated sensitivity by a factor of $\O{d^{1+0.5z}\log^{z-1}\Delta}$, in which case sensitivity sampling in the second pass could require sampling $\Omega(\log^{z-1}\Delta)$ points. 
Unfortunately, even for $z=2$, this no longer gives $\O{\log(n\Delta)}$ total bits of space if we store each sampled point using $\O{\log(n\Delta)}$ bits of space. 

To overcome this issue, we store an approximation of its \emph{offset} from each of the centers instead of storing an explicit representation of each point. 
That is, for a point $x$ and a set $C'$ of $\O{k}$ centers, suppose $c'(x)$ is the closest center of $C'$ to $x$. 
Let the offset from $x$ to $c'(x)$ be defined by $y'=x-c'(x)$. 
Suppose we round each coordinate of $y'$ to a power of $\left(1+\poly\left(\eps,\frac{1}{d},\frac{1}{\log\Delta}\right)\right)$ to form a vector $y$. 

Let $Y$ be the set of all points of $X$ rounded in this manner. 
We show that the cost of any clustering $C$ on $Y$ is a $(1+\eps)$-approximation to the cost of the clustering $C$ and $X$, so it suffices to instead consider $(k,z)$-clustering on $Y$. 
The main insight from considering $Y$ is that due to the points all being rounded to a power of $\left(1+\poly\left(\eps,\frac{1}{d},\frac{1}{\log\Delta}\right)\right)$ away from a point in $C'$, then there is a compact representation of each point in $Y$ that does not require $\Omega(d\log\Delta)$ bits to represent. 
In fact, to represent each point $y$, it suffices to store the identity of $c'(x)$ as well as the \emph{exponents} of the offsets, which only requires $\O{d\log\frac{\log\Delta}{\eps}}$ bits per sample. 
Hence for $k$-means clustering, the algorithm still uses $o(\log^2(n\Delta))$ total bits of space. 

\subsection{Preliminaries}
\seclab{sec:prelims}
For an integer $n>0$, we use $[n]$ to denote the set $\{1,\ldots,n\}$. 
We use $\poly(n)$ to denote a fixed polynomial in $n$ and $\polylog(n)$ to denote $\poly(\log n)$. 
If an event occurs with probability at least $1-\frac{1}{\poly(n)}$, we say the event occurs with high probability. 

For vectors $x,y\in\mathbb{R}^d$, we use $\dist(x,y)$ to denote the Euclidean distance $\|x-y\|_2$, so that $\|x-y\|_2^2=\sum_{i=1}^d(x_i-y_i)^2$. 
More generally, we define the $L_z$ norm of $x$ by $\|x\|_z$ so that $\|x\|_z^z=\sum_{i=1}^d x_i^z$. 
For a set $S$, we use $\dist(x,S)$ to denote $\min_{y\in S}\dist(x,y)$ and similarly, for sets $X$ and $S$, we define $\dist(X,S)=\min_{x\in X,y\in S}\dist(x,y)$. 
For a fixed $z\ge 1$ and sets $X,C\subset\mathbb{R}^d$ with $X=\{x_1,\ldots,x_n\}$ we use $\Cost(X,C)$ to denote $\sum_{i=1}^n\dist(x_i,C)^z$. 

\begin{fact}[Generalized triangle inequality]
\factlab{fact:triangle}
For any $z\ge 1$ and $x,y,z\in\mathbb{R}^d$, we have
\[\dist(x,y)^z\le 2^{z-1}(\dist(x,w)^z+\dist(w,y)^z).\]
\end{fact}

Let $x,y\in\mathbb{R}^n$ with $n=[\Delta]^d$ be two vectors such that $\|x\|_1=\|y\|_1$. 
Let $F(x,y)$ be the family of functions $f:[\Delta]^d\times[\Delta]^d\to\mathbb{R}^{\ge 0}$ so that for any $i\in[\Delta]^d$, we have $\sum_{j\in[\Delta]^d} f(i,j)=x_i$ and for any $j\in[\Delta]^d$, we have $\sum_{i\in[\Delta]_d} f(i,j)=y_j$. 
Then we define the earth mover distance (EMD) between $x$ and $y$ by
\[\EMD(x,y)=\min_{f\in F}\sum_{i,j\in[\Delta]^d}f(i,j)\|i-j\|_2.\]
For general vectors $x,y\in\mathbb{R}^n$, we define
\[\EMD(x,y)=\inf_{x'\preceq x, y'\preceq y, \|x'\|_1=\|y'\|_=1}\EMD(x',y')+d\Delta(\|x-x'\|_1+\|y-y'\|_1).\]
More generally, we define the quantity $\WASSZZ(x,y)$ for $x,y\in\mathbb{R}^n$ with $n=[\Delta]^d$ and $\|x\|_1=\|y\|_1$ by
\[\WASSZZ(x,y)=\min_{f\in F}\sum_{i,j\in[\Delta]^d}f(i,j)\|i-j\|_2^z\]
and for general vectors $x,y\in\mathbb{R}^n$ by
\[\WASSZZ(x,y)=\inf_{x'\preceq x, y'\preceq y, \|x'\|_1=\|y'\|_=1}\WASSZZ(x',y')+d^z\Delta^z(\|x-x'\|_1+\|y-y'\|_1).\]
We remark that the quantity $\WASSZZ(x,y)$ corresponds to the $z$-th power of the Wasserstein-$z$ distance between $x$ and $y$; we use $\WASSD(x,y)$ to denote the Wasserstein-$z$ distance between $x$ and $y$, in general omitting the dependency on $z$ when it is clear from context. 

We also define $x^+=\frac{|x|+x}{2}$ to be the vector containing the positive entries of $x$ and $x^-=x-x^+$ to be the vector containing the negative entries of $x$. 
Then we define 
\[\|x\|_{\EMD}=\EMD(x^+,x^-).\]

We now formally define the notion of a (strong) coreset. 
\begin{definition}[Coreset]
\deflab{def:coreset}
Given an approximation parameter $\eps>0$, and a set $X$ of points $x_1,\ldots,x_n\in\mathbb{R}^d$ with distance function $\dist$, a \emph{coreset} for $(k,z)$ clustering is a subset $S$ of weighted points of $X$ with weight function $w$ such that for any set $C$ of $k$ points, we have
\[(1-\eps)\sum_{t=1}^n\dist(x_t,C)^z\le\sum_{q\in S}w(q)\dist(q,S)^z\le(1+\eps)\sum_{t=1}^n\dist(x_t,C)^z.\]
\end{definition}

Currently, the state-of-the-art coreset construction for $(k,z)$-clustering is the following:
\begin{theorem}\cite{Cohen-AddadLSS22}
\thmlab{thm:offline:coreset:size}
Given an accuracy parameter $\eps\in(0,1)$, there exists a coreset construction for $(k,z)$-clustering that samples $\tO{\frac{k}{\eps^2}2^{z\log z}\cdot\min\left(\frac{1}{\eps^z},k\right)}$ weighted points. 
\end{theorem}

Thus, merge-and-reduce on a stream of length $n$ using the coreset construction of \cite{Cohen-AddadLSS22} in \thmref{thm:offline:coreset:size} offers the following guarantee:
\begin{lemma}
\lemlab{lem:space:mr}
There exists a one-pass streaming algorithm on insertion-only data streams of length $n$ that, with probability at least $0.99$, simultaneously outputs a $(1+\eps)$-approximation to $(k,z)$-clustering at all times of the stream, using $\tO{\frac{k}{\eps^2}2^{z\log z}\cdot\left(\frac{1}{\eps^z},k\right)}\cdot\polylog(n)$ bits of space. 
\end{lemma}

We recall the following property of the Johnson-Lindenstrauss (JL) transformation. 
\begin{theorem}[Johnson-Lindenstrauss lemma]
\cite{johnson1984extensions}
\thmlab{thm:jl}
Let $X\subset\mathbb{R}^d$ be a set of $n$ points and $m=\O{\frac{1}{\eps^2}\log n}$. 
There exists a family of random linear maps $\frakM:\mathbb{R}^d\to\mathbb{R}^m$ such that with probability at least $0.99$ over the choice of $\frakM$, we simultaneously have for all $x,y\in X$,
\[(1-\eps)\|x-y\|_2\le\|\frakM x-\frakM y\|_2\le(1+\eps)\|x-y\|_2.\]
\end{theorem}

We also recall the following concentration inequalities:
\begin{theorem}[Bernstein's concentration inequality]
\cite{bernstein1927theory}
\thmlab{thm:bernstein}
Let $X_1,\ldots,X_m$ be independent random variables with $\Ex{X_i}<\infty$ and $X_i\ge0$ for all $i\in[m]$. 
Let $X=\sum_{i=1}^m X_i$ and let $\gamma>0$. 
Then
\[\PPr{X\le\Ex{X}-\gamma}\le\exp\left(\frac{-\gamma^2}{2\sum_i\Ex{X_i^2}}\right).\]
Moreover, if $X_i-\Ex{X_i}\le\Delta$ for all $i\in[m]$, then for $\sigma_i^2:=\Ex{X_i^2}-\Ex{X_i}^2$,
\[\PPr{X\ge\Ex{X}+\gamma}\le\exp\left(\frac{-\gamma^2}{2\sum_i\sigma_i^2+2\gamma\Delta/3}\right).\]
\end{theorem}

\begin{theorem}[Freedman's inequality]
\cite{Freedman75}
\thmlab{thm:scalar:freedman}
Suppose $Y_0,Y_1,\ldots,Y_n$ is a scalar martingale with difference sequence $X_1,\ldots,X_n$. 
Specifically, we initiate $Y_0=0$ and set $Y_i=Y_{i-1}+X_i$ for all $i\in[n[$
Let $R\ge|X_t|$ for all $t\in[n]$ with high probability. 
We define the predictable quadratic variation process of the martingale by $w_k:= \sum_{t=1}^k\underset{t-1}{\mathbb{E}}\left[X_t^2\right]$, for $k\in[n]$. 
Then for all $\eps\ge 0$ and $\sigma^2 > 0$, and every $k \in [n]$, 
\[\PPr{\max_{t \in [k]} |Y_t|>\eps\text{ and } w_k \le \sigma^2}\le 2\exp\left(-\frac{\eps^2/2}{\sigma^2 + R\eps/3} \right).\]
\end{theorem}

\section{EMD Sketch}
\seclab{sec:emd:sketch}
We first describe the embedding of $\EMD$ on $[\Delta]^d$ into $L_1$. 
Suppose without loss of generality that $\Delta=2^\ell$ for some non-negative integer $\ell$. 

Given $s=(s_1,\ldots,s_d)\in\mathbb{Z}^d$ and $t\in\{0,1,\ldots,\ell\}$, we define the linear map $G_{s,t}$ from $\mathbb{R}^{[\Delta]^d}$ into $L_1$ by first defining the grid $\calG_{s,t}$ over $\mathbb{Z}^d$ with side length $2^t$, so that $s=(s_1,\ldots,s_d)$ lies on one of the corners of the grid. 
Then for $\mu\in\mathbb{R}^{[\Delta]^d}$, we define $G_{s,t}\mu$ as the frequency vector over the hypercubes of the grid $\calG_{s,t}$ that counts the total mass in each hypercube. 
To avoid ambiguity, we say that each cell of a grid of length $2^t$ has closed boundaries on one side and open boundaries on the other side, e.g., a cell that contains a point $(x_1,\ldots,x_d)$ may contain points $(y_1,\ldots,y_d)$ where $y_i<x_i+2^k$, but does not contain any points $(y_1,\ldots,y_d)$ where $y_i\ge x_i+2^k$. 
Then we define the mapping $G_s\mu$ to be the concatenation of the vectors 
\[G_s\mu=(G_{s,0}\mu)\circ(2\cdot G_{s,1}\mu)\circ\ldots\circ(2^t\cdot G_{s,t}\mu)\circ\ldots\circ(2^\ell\cdot G_{s,\ell}\mu).\]

We first require the following structural property:
The following is analogous to Claim 6.1 in \cite{BackursIRW16}, generalized to dimension $d$. 
\begin{lemma}
\lemlab{lem:markov:tail:entropy}
Let $X_1,\ldots,X_n$ be non-negative random variables not necessarily independent, such that for every $i\in[n]$ and $d,t>0$, 
\[\PPr{X_i\ge t}\le\frac{\gamma d}{t},\]
for an absolute constant $\gamma>0$. 
Let $S=\sum_{i\in[n]}\alpha_i X_i$ for a set of coefficients $\alpha_1,\ldots,\alpha_n\ge 0$ with $\sum_{i\in[n]}\alpha_i=1$. 
Then there exists an absolute constant $B>0$ such that for every $\delta\in(0,1)$,
\[\PPr{S\ge\frac{B\gamma d}{\delta}\left(H(\alpha)+\log\frac{2\gamma d}{\delta}\right)}\le\delta,\]
where $H(\alpha)$ is the entropy of the distribution over $[n]$ induced by $\alpha_1,\ldots,\alpha_n$. 
\end{lemma}
\begin{proof}
Let $T_1,\ldots,T_n\ge 0$ be non-negative parameters to be fixed. 
We define $\calE$ to be the event that $X_i\le T_i$ simultaneously for all $i\in[n]$. 
By the assumption that $\PPr{X_i\ge T_i}\le\frac{\gamma d}{T_i}$ for all $i\in[n]$ and a union bound, we have that
\[\PPr{\neg\calE}\le\sum_{i\in[n]}\frac{\gamma d}{T_i}.\]
Moreover, there exists an absolute constant $B$ such that $\Ex{X_i\mid\calE}\le B\gamma d\log T_i$. 
Hence by Markov's inequality, 
\[\PPr{S\le\frac{B\gamma d}{\delta}\left(\sum_{i\in[n]}\alpha_i\log T_i\right)\mid\calE}\ge 1-\frac{\delta}{2}.\]
By setting $T_i=\frac{2\gamma d}{\alpha_i\delta}$, we have that
\[\PPr{\neg\calE}\le\sum_{i\in[n]}\frac{\gamma d}{T_i}\le\sum_{i\in[n]}\frac{\alpha_i\delta}{2}=\frac{\delta}{2}\]
and
\[\sum_{i\in[n]}\alpha_i\log T_i=H(\alpha)+\log\frac{2\gamma d}{\delta}.\]
\end{proof}

We next generalize a statement lower bounding the contraction of the EMD embedding into $L_1$ by \cite{BackursIRW16,indyk120oct} from $[\Delta]^2$ to $[\Delta]^d$. 
Namely, \cite{BackursIRW16,indyk120oct} showed that for $[\Delta]^2$, the contraction is upper bounded by a constant amount. 
Here we show that for $[\Delta]^d$, the contraction is at most $\O{\sqrt{d}}$. 
\begin{lemma}
\lemlab{lem:gs:contraction}
For every $\mu,\nu\in\EMD_{[\Delta]^d}$ and every $s=(s_1,\ldots,s_d)\in\mathbb{Z}^d$, it holds that $\|\mu-\nu\|_{\EMD}\le\O{\sqrt{d}}\cdot\|G_s(\mu-\nu)\|_1$.
\end{lemma}
\begin{proof}
Suppose without loss of generality that $\mu$ and $\nu$ have non-intersecting support, since otherwise the intersection can be subtracted from both $\mu$ and $\nu$. 
Moreover, suppose without loss of generality that $\mu$ and $\nu$ are probability measures, i.e., $\|\mu\|_1=\|\nu\|_1=1$. 

We first consider the matching induced by pairing mass contained in the same cells of grids, starting with $\calG_{s,0}$ and then iterating through $\calG_{s,1},\ldots$. 
Since $\mu$ and $\nu$ have non-intersecting support and all points in $[\Delta]^d$ fall in separate cells in $\calG_{s,0}$, then we have $\|G_{s,0}(\mu-\nu)\|_1=2$. 
Then there is at most $1-\frac{1}{2}\|G_{s,1}(\mu-\nu)\|_1$ mass that can be matched together in $\calG_{s,1}$, which induces cost at most $2\sqrt{d}(1-\frac{1}{2}\|G_{s,1}(\mu-\nu)\|_1)$. 

Next, at most $\frac{1}{2}\|G_{s,1}(\mu-\nu)\|_1-\frac{1}{2}\|G_{s,2}(\mu-\nu)\|_1$ mass can be matched together in $\calG_{s,2}$, inducing cost at most 
\[4\sqrt{d}\left(\frac{1}{2}\|G_{s,1}(\mu-\nu)\|_1-\frac{1}{2}\|G_{s,2}(\mu-\nu)\|_1\right).\] 
More generally, the cost induced by grid $\calG_{s,i}$ can be at most
\[2^i\sqrt{d}\left(\frac{1}{2}\|G_{s,1}(\mu-\nu)\|_1-\frac{1}{2}\|G_{s,2}(\mu-\nu)\|_1\right).\]

Hence, we have
\begin{align*}
\|\mu-\nu\|_{\EMD}&\le\frac{\sqrt{d}}{2}\sum_{i=0}^{\ell-1} 2^i(\|G_{s,i}(\mu-\nu)\|_1-\|G_{s,i+1}(\mu-\nu)\|_1)\\
&\le\frac{\sqrt{d}}{2}\|G_s(\mu-\nu)\|_1,
\end{align*}
since
\[G_s\mu=(G_{s,0})\mu\circ(2\cdot G_{s,1}\mu)\circ\ldots\circ(2^t\cdot G_{s,t}\mu)\circ\ldots\circ(2^\ell\cdot G_{s,\ell}\mu).\]
\end{proof}

Before lower bounding the dilation of the EMD embedding, we first define the following concept to quantify the granularity at which the randomly shifted grid intersects a specific $L_1$ ball. 
\begin{definition}
Let $x\in\mathbb{R}^d$, $R\in(0,2\Delta]$, and let $s=(s_1,\ldots,s_d)\in[\Delta]^d$ be sampled uniformly at random. 
Let
\[\calA_{x,R}(s):=\frac{\min\left(\{2^t\,\mid\,\calG_{s,t}\text{ does not cut }B_{L_1}(x,R)\}\right)}{R}.\]
\end{definition}
Intuitively, $\calA_{x,R}(s)$ is the side length of the finest grid among the $\ell+1$ grids $\calG_{s,0},\ldots,\calG_{s,\ell}$ that does not intersect the $L_1$ ball $B_{L_1}(x,R)$. 

We next generalize Lemma 6.2 in \cite{BackursIRW16} to $d$-dimensional space. 
Intuitively, the following statement shows that the probability that coarser grids, i.e., grids containing cells with larger side lengths, are linearly less likely to intersect with a fixed $L_1$ ball. 
\begin{lemma}
\lemlab{lem:grid:cut:tail}
There exists a universal constant $\gamma>0$ such that for every $x\in\mathbb{R}^d$, $R\in(0,2\Delta]$ and $T>0$,
\[\PPPr{s\in[\Delta]^d}{\calA_{x,R}(s)\ge T}\le\frac{\gamma d}{T}.\]
\end{lemma}
\begin{proof}
Observe that by definition, $\calA_{x,R}(s)\ge T$ if and only if the grid with the largest side length less than $RT$ intersects with the ball $B_{L_1^d}(x,R)$. 
Note that for each of the $d$ dimensions, the probability that $B_{L_1^d}(x,R)$ intersects with the grid of length $RT$ is at most $\frac{\gamma}{T}$ for some absolute constant $\gamma>0$. 
Thus by a union bound, we have that $\PPPr{s\in[\Delta]^d}{\calA_{x,R}(s)\ge T}\le\frac{\gamma d}{T}$.
\end{proof}

Similarly, we generalize Lemma 6.3 in \cite{BackursIRW16} to $[\Delta]^d$. 
Intuitively, the following statement relates the embedding of the difference of two elementary vectors to the quantity $\calA_{u,R}(s)$ representing the side length of the finest grid that does not cut the $L_1$ ball $B_{L_1}(u,R)$. 
\begin{lemma}
\lemlab{lem:grid:evec:diff}
Let $x,y\in[\Delta]^d$ and let $s=(s_1,\ldots,s_d)\in[\Delta]^d$ be sampled uniformly at random. 
Let $e_x$ and $e_y$ be the basis vectors corresponding to $x$ and $y$ respectively. 
Then there exists an absolute constant $\gamma>0$ such that $\|G_s(e_x-e_y)\|_1\ge\gamma\cdot R\cdot\calA_{u,R}(s)$ for every $u\in\mathbb{R}^2$ and $R\in(0,2\Delta]$ such that the ball $B_{L_1}(u,R)$ contains both $x$ and $y$.
\end{lemma}
\begin{proof}
Observe that by the definition of $\calA_{u,R}$, no grid with side length at least $R\cdot\calA_{u,R}(s)$ contributes to $\|G_s(e_x-e_y)\|_1$. 
On the other hand, all grids with side length less than $R\cdot\calA_{u,R}(s)$ contribute a geometric series toward $\|G_s(e_x-e_y)\|_1$ with total sum at most $\gamma\cdot R\cdot\calA_{u,R}(s)$ for some constant $\gamma>0$. 
\end{proof}

The following statement is the generalization of Claim 6.2 in \cite{BackursIRW16} to $[\Delta]^d$. 
The main point is to we show that we can upper bound $\|G_s\mu\|_1$ in a way that the assumptions of \lemref{lem:markov:tail:entropy} can be applied. 
In the following statement, we say a random variable $X$ statistically dominates a random variable $Y$ if for all real numbers $t$, we have $\PPr{X\ge t}\ge\PPr{Y\ge t}$. 
\begin{lemma}
\lemlab{lem:sum:markov:tail}
Let $\mu$ and $\nu$ be two probability measures over $[\Delta]^d$ and suppose the optimal transportation of $\mu$ to $\nu$ consists of moving mass $w_i$ from the point $x_i\in[\Delta]^d$ to the point $y_i\in[\Delta]^d$, for all $i\in[k]$. 
Let $\{B_j=B_{L_1}(u_j,R_j)\}_{j=1}^q$ be a collection of $L_1$ balls in $d$-dimensional space such that for every $i\in[k]$, there exists $f(i)\in[q]$ such that both $x_i$ and $y_i$ are contained within $B_{f(i)}$. 
For each $j\in[q]$, let 
\[\widetilde{w_j}=\sum_{i:f(i)=j}w_i.\]
Let $s=(s_1,\ldots,s_d)\in[\Delta]^d$ be sampled uniformly at random. 
Then there exists an absolute constant $\gamma>9$ such that the random variable
\[\|G_s(\mu-\nu)\|_1\le\sum_{i\in[k]}w_i\|G_s(e_{x_i}-e_{y_i})\|_1\]
is statistically dominated by $S=\sum_{j\in[k]}\widetilde{w_j}\cdot R_j\cdot X_i$ for some non-negative and not necessarily independent random variables $X_1,\ldots,X_k$ with the property for every $i$ and $t>0$,
\[\PPr{X_i\ge T}\le\frac{\gamma d}{T}.\]
\end{lemma}
\begin{proof}
The claim immediately follows by applying \lemref{lem:grid:cut:tail}, \lemref{lem:grid:evec:diff}, and the triangle inequality. 
\end{proof}

The following statement is the generalization of Claim 6.3 in \cite{BackursIRW16} to $[\Delta]^d$. 
Intuitively, we upper bound the $L_1$ norm of the embedding $G_s(\mu-\nu)$ in terms of the entropy $H(\alpha)$ of the coefficient vector $\alpha$. 
\begin{lemma}
\lemlab{lem:gs:one:final}
Assuming the notation and conditions from \lemref{lem:sum:markov:tail},
\[\PPr{\|G_s(\mu-\nu)\|_1\le \O{d}\cdot H(\alpha)\cdot T}\ge 0.999,\]
where $T=\sum_{j=1}^q\widetilde{w_j}R_j=\sum_{i=1}^k w_iR^{f(i)}$ and $\alpha_j=\frac{\widetilde{w_j}R_j}{T}$ for $j\in[q]$. 
\end{lemma}
\begin{proof}
The claim immediately holds by \lemref{lem:sum:markov:tail} and \lemref{lem:markov:tail:entropy}. 
\end{proof}

We now upper bound the dilation of the EMD embedding into $L_1$. 
Namely, we show that the embedding into $L_1$ can be at most a multiplicative $\O{d(\log k+\log\log\Delta)}$ factor larger than the EMD. 
\begin{lemma}
\lemlab{lem:gs:dilation}
Suppose that $\mu$ and $\nu$ are two probability measures over $[\Delta]^d$ and in addition, $\mu$ has support of size at most $k$. 
Then
\[\PPr{\|G_s(\mu-\nu)\|_1\le \O{d(\log k+\log\log\Delta)}\cdot \|\mu-\nu\|_{\EMD}}\ge 0.999.\]
\end{lemma}
\begin{proof}
Let the support of $\mu$ be $x_1,\ldots,x_k\in[\Delta]^d$. 
Let $\calB$ be the family of $\O{k\log\Delta}$ balls $\{B(x_i,2^{j-1})\}_{i\in[k],j\in[2+\log\Delta]}$. 

Consider the optimal transportation from $\mu$ to $\nu$. 
Each edge of length $L$ in this transportation is contained in a ball of radius $\O{L}$. 
Hence, by taking $T\le\O{1}\cdot\|\mu-\nu\|_{\EMD}$ in \lemref{lem:gs:one:final}, we have
\[\PPr{\|G_s(\mu-\nu)\|_1\le \O{t\log t}\cdot H(\alpha)\cdot \|\mu-\nu\|_{\EMD}}\le\frac{1}{t},\]
and it remains to upper bound $H(\alpha)$. 
Since the support of $\alpha$ has size at most $\O{k\log\Delta}$, then $H(\alpha)\le\O{\log(k\log\Delta)}\le\O{\log k+\log\log\Delta}$. 
Therefore,
\[\PPr{\|G_s(\mu-\nu)\|_1\le \O{d}\cdot\O{\log k+\log\log\Delta}\cdot \|\mu-\nu\|_{\EMD}}\ge 0.999.\]
\end{proof}

In summary, we have the following bounds on the contraction and the dilation of the EMD embedding into $L_1$. 
\begin{theorem}
\thmlab{thm:gs:contraction:dilation}
For every $\mu,\nu\in\EMD_{[\Delta]^d}$ and every $s=(s_1,\ldots,s_d)\in\mathbb{Z}^d$, it holds that 
\[\|\mu-\nu\|_{\EMD}\le\O{\sqrt{d}}\cdot\|G_s(\mu-\nu)\|_1.\] 
In addition, if $\mu$ and $\nu$ are two probability measures over $[\Delta]^d$ and $\mu$ has support of size at most $k$. 
Then
\[\PPr{\|G_s(\mu-\nu)\|_1\le \O{d}\cdot\O{\log k+\log\log\Delta}\cdot \|\mu-\nu\|_{\EMD}}\ge 0.999,\]
\end{theorem}
\begin{proof}
The first part follows from \lemref{lem:gs:contraction} and the second part follows from \lemref{lem:gs:dilation}. 
\end{proof}
We now require simple generalizations of statements from \cite{BackursIRW16} to the Wasserstein-$z$ distance, as follows:
\begin{lemma}[Claim 5.1 in~\cite{BackursIRW16}]
\lemlab{lem:doubling:flow:upper}
Let $\mu$ be a fixed $k$-sparse probability measure over $[\Delta]^d$ and let $R>0$ be a fixed radius. 
For every $k$-sparse $\mu'$ such that $\WASSD(\mu, \mu')\le R$, there exists an optimal flow between $\mu$ and $\mu'$ that is supported on at most $2k$ pairs of points. 
\end{lemma}

\begin{algorithm}[!htb]
\caption{$\WASSD$ Net Construction}
\alglab{alg:emd:net}
\begin{algorithmic}[1]
\Require{Radius $R>0$, parameters $d$, $\Delta$, $k$}
\Ensure{Net $\calM$ on $k$-sparse probability measures with respect to $\WASSD$}
\State{$m_0\gets\frac{R}{100d\Delta k}$, $\calM\gets\emptyset$}
\For{$c:\supp\mu\to\mathbb{Z}_{>0}$ such that $\sum_{(x_1,\ldots,x_d)\in\supp\mu}c(x_1,\ldots,x_d)\le 2k$}
\State{$\calI\gets\{(i,x_1,\ldots,x_d)\,\mid\,(x_1,\ldots,x_d)\in\supp\mu, i\in[c(x_1,\ldots,x_d)]$}
\For{$\ell:\calI\to\{1,1.01,1.01^2,\ldots,2\Delta\}$}
\For{$(i,x_1,\ldots,x_d)\in\calI$ and for all $p(i,x_1,\ldots,x_d)\in\BuildNet((x_1,\ldots,x_d),\ell(i,x_1,\ldots,x_d))$}
\For{$m:\calI\to\{0,m_0,1.01m_0,1.01^2m_0,\ldots,\min(1,R)\}$}
\If{for every $(x_1,\ldots,x_d)\in\supp\mu$, it holds that $\sum_{i:(i,x_1,\ldots,x_d)\in\calI}m(i,x_1,\ldots,x_d)\le\mu(x_1,\ldots,x_d)$}
\State{Let $\mu'$ be a measure over $[\Delta]^d$ that is identically zero}
\For{$(x_1,\ldots,x_d)\in\supp\mu$}
\State{$s\gets 0$}
\For{$i:(i,x_1,\ldots,x_d)\in\calI$}
\State{$s\gets s+m(i,x_1,\ldots,x_d)$}
\State{$\mu'(p(i,x_1,\ldots,x_d))\gets\mu'(p(i,x_1,\ldots,x_d))+m(i,x_1,\ldots,x_d)$}
\EndFor
\State{$\mu'(x_1,\ldots,x_d)\gets\mu'(x_1,\ldots,x_d)+\mu(x_1,\ldots,x_d)-s$}
\EndFor
\State{$\calM\gets\calM\cup\{\mu'\}$}
\EndIf
\EndFor
\EndFor
\EndFor
\EndFor
\end{algorithmic}
\end{algorithm}

\begin{lemma}
\lemlab{lem:wass:net}
The doubling dimension of the set of $k$-sparse probability measures over $[\Delta]^d$ under $\WASSD$ is $\O{dk\log\log\Delta}$. 
\end{lemma}
\begin{proof}
Let $\mu$ be a fixed $k$-sparse probability measure over $[\Delta]^d$ and let $R>0$ be a fixed radius. 
Let $B_{\WASSD}(\mu,R)$ denote the set of points $\nu$ such that $\WASSD(\mu,\nu)\le R$. 
We show that $B_{\WASSD}(\mu,R)$ can be covered with $(\log\Delta)^{\O{kd}}$ $\WASSD$-balls with radius $\frac{R}{2}$ that are centered at $k$-sparse measures. 
Note that it suffices to cover $B_{\WASSD}(\mu,R)$ with $(\log\Delta)^{\O{kd}}$ $\WASSD$-balls with radius $\frac{R}{4}$ that are centered at arbitrary probability measures. 

In \algref{alg:emd:net}, we construct a set of measures $\calM$ that form the centers of balls with radius $\frac{R}{4}$ whose union together covers $B_{\WASSD}(\mu,R)$. 
The high-level approach is to first enumerate over all possible topologies of the optimal flow, then enumerate over all possible lengths of the corresponding edges, then enumerate over all possible supports, and finally, enumerate over all possible masses that are transported over the edges. 

We first assume access to a subroutine $\BuildNet(p,r)$ that returns a $\frac{r}{100}$-net of $B_{L_z}(p,r)\cap[\Delta]^d$, i.e., a net over the $L_z$ ball of radius $r$ centered that $p$ restricted to the points in $[\Delta]^d$. 
It follows that $|\calM|\le(\log\Delta)^{\O{dk}}$ and thus, the runtime of \algref{alg:emd:net} is also at most $(\log\Delta)^{\O{dk}}$. 
It remains to show that for every $k$-sparse probability measure $\mu'$ such that $\WASSD(\mu,\mu')\le R$, there exists $\mu''\in\calM$ with $\WASSD(\mu'',\mu')\le\frac{R}{4}$.

By \lemref{lem:doubling:flow:upper}, there is an optimal flow between $\mu$ and $\mu'$ containing at most $2k$ edges. 
Thus, in the outer for loop of \algref{alg:emd:net}, there exists at least one guess of $c(x_1,\ldots,x_d)$ that indeed corresponds to the number of outgoing edges from $(x_1,\ldots,x_d)$ in the optimal flow. 
In the second for loop of \algref{alg:emd:net} that guesses the flow for each edge, there is at least one choice that guesses all flows within a multiplicative factor of $1.01$. 
Hence, there exists a measure $\tilde{\mu}$ such that
\begin{enumerate}
\item
$\supp(\tilde{\mu})\in\supp(\mu)\cup\{p(i,x_1,\ldots,x_d)\}_{(i,x_1,\ldots,x_d)\in\calI}$
\item
$\WASSD(\mu',\tilde{\mu})\le\frac{R}{50}$
\item
There exists a flow between $\mu$ and $\tilde{\mu}$ with cost at most $1.02R$ and transports mass from each $(x_1,\ldots,x_d)\in\supp(\mu)$ to some point $\{(x_1,\ldots,x_d)\}\cup\{p(i,x_1,\ldots,x_d)\}_{(i,x_1,\ldots,x_d)\in\calI}$.
\end{enumerate}

Consequently, \algref{alg:emd:net} will guess $\supp(\tilde{\mu})$ and then the measure at the support. 
We show that some guess $\mu''$ by \algref{alg:emd:net} will satisfy $\supp(\mu'')\subseteq\supp(\tilde{\mu})$ and $\WASSD(\tilde{\mu},\mu'')\le\frac{R}{25}$. 

To that end, we first define $\mu''$. 
Consider a fixed $(x_1,\ldots,x_d)\in\supp(\mu)$ and the corresponding multi-set $\{(x_1,\ldots,x_d)\}\cup\{p(i,x_1,\ldots,x_d)\,\mid\,(i,x_1,\ldots,x_d)\in\calI\}$. 
Round down the mass of $\tilde{\mu}$ at the coordinates $\{p(i,x_1,\ldots,x_d)\}$ to the closest element of $\{0,m_0,1.01m_0,1.01^2m_0,\ldots,\min(1,R)\}$ and let $\mu''$ be the resulting measure. 
We set $\mu''((x_1,\ldots,x_d))=\sum_{i:(i,x_1,\ldots,x_d)\in\calI)}(\tilde{\mu}(p(i,x_1,\ldots,x_d))-\mu''(p(i,x_1,\ldots,x_d)))$ and emphasize that $\mu''$ is one of the measures guessed by \algref{alg:emd:net}. 

We now show $\WASSD(\tilde{\mu},\mu'')\le\frac{R}{25}$, by upper bounding the following two terms:
\begin{enumerate}
\item 
We first upper bound the contribution from $(i,x_1,\ldots,x_d)\in\calI$ for which $\tilde{\mu}(p(i,x_1,\ldots,x_d))<m_0$, so that $\mu''(p(i,x_1,\ldots,x_d))=0$. 
There exist at most $2k$ elements $(i,x_1,\ldots,x_d)\in\calI$ and these elements can be rerouted with cost at most $2\Delta \sqrt{d}km_0\le\frac{R}{50}$. 
\item
We next upper bound the contribution from $(i,x_1,\ldots,x_d)\in\calI$ for which $\tilde{\mu}(p(i,x_1,\ldots,x_d))\ge m_0$, so that $\mu''(p(i,x_1,\ldots,x_d))=0$ is within a multiplicative $1.01$ factor of $\tilde{\mu}(p(i,x_1,\ldots,x_d))=0$. 
Hence, the total contribution of such elements is at most $0.01\WASSD(\mu,\tilde{\mu})\le\frac{R}{50}$. 
\end{enumerate}
Therefore, $\WASSD(\mu'-\mu'')\le\WASSD(\mu'-\tilde{\mu})+\WASSD(\tilde{\mu}-\mu'')\le\frac{R}{50}+\frac{R}{50}+\frac{R}{50}<\frac{R}{4}$. 
\end{proof}
We can use the results for the doubling dimension of the set of $k$-sparse Wasserstein-$z$ (and in particular, earth mover distance for $z=1$) probability measures as follows. 
\begin{lemma}[Lemma 3.1 in \cite{BackursIRW16}]
\lemlab{lem:emd:recovery}
Let $\calM=(X,\rho)$ be a $K$-quasi-metric space and $Y\subseteq X$ have doubling dimension $d$. 
Suppose there exists a sketch of size $s$ that approximates distances between points of $X$ and $Y$ with distortion $D$ with probability at least $\frac{2}{3}$. 
Then for every $\eps\in\left(0,\frac{1}{2}\right)$, $\lambda\in(0,\Lambda)$, and $y_0\in Y$, one can sketch points of $X$ with sketch size $\O{s(d\log(DK/\eps)+\log\log(\Lambda/\lambda))}$, so that from this sketch for $x\in X$ with $\rho(x,y_0)\le\Lambda$, with probability at least $\frac{2}{3}$, we can recover a point $y'\in Y$ such that
\[\rho(x,y')\le\max((1+\eps)DK\cdot\rho(x,Y),\lambda).\]
\end{lemma}
\lemref{lem:emd:recovery} avoids a brute force search over the entire metric space by cleverly utilizing the doubling dimension to perform an efficient top-down search. 
In particular, it first partitions the search space into balls of radius $R$ and identifies the ball $B$ that contains the smallest distance from $X$. 
It then partitions $B$ into smaller balls of radius $\frac{R}{2}$ and iterates to identify $y'$. 

\section{Online Sensitivity Sampling and Insertion-Only Streams}
In this section, we present our simple algorithm for $(k,z)$-clustering in the insertion-only model. 
Along the way, we first show in \secref{sec:online:sens:samp} important properties of online sensitivity sampling that we believe could be of independent interest. 
Namely, we show that online sensitivity sampling gives a $(1+\eps)$-coreset to the underlying dataset for $(k,z)$-clustering. 
We also upper bound the sum of the online sensitivities, which allows us to upper bound the number of samples procured by online sensitivity sampling. 
Finally, we give our algorithm for $(k,z)$-clustering on insertion-only streams in \secref{sec:kz:insertion}. 

\subsection{Online Sensitivity Sampling}
\seclab{sec:online:sens:samp}
In this section, we show that online sensitivity sampling can be used to achieve a $(1+\eps)$-coreset for $(k,z)$-clustering. 
We show that an approximately optimal clustering can be used to compute approximations of the online sensitivities. 
We remark that the same statement can also be used to compute approximations of the sensitivities of query points. 
\begin{lemma}
\lemlab{lem:opt:approx:sens}
Let $S\subset\mathbb{R}^d$ with $|S|=k$ satisfy $\Cost(X,S)\le\gamma\OPT$ for some $\gamma\ge 1$, where $\OPT$ denotes the cost of an optimal $(k,z)$-clustering. 
Let $Z\le\Cost(X,S)\le\beta Z$. 
Let $\pi:X\to S$ be the mapping from $X$ to $S$ that induces the clustering cost $\Cost(X,S)$ and let $Z'\le\Cost(C,\pi(X))\le\beta'Z'$. 
Then for any $C\subset\mathbb{R}^d$ with $|C|=k$, 
\[\frac{\Cost(x,C)}{\Cost(X,C)}\le\frac{2^z\cdot 4\gamma\Cost(x,C)}{Z'+Z}\le\frac{2^z\cdot4\beta\beta'\gamma\Cost(x,C)}{\Cost(X,C)}.\]
\end{lemma}
\begin{proof}
Let $S\subset\mathbb{R}^d$ with $|S|=k$ satisfy $\Cost(X,S)\le\gamma\OPT$ for some $\gamma\ge 1$. 
Let $\pi:X\to S$ be the mapping from $X$ to $S$ that induces the clustering cost $\Cost(X,S)$. 
Since $\gamma\Cost(X,C)\ge\gamma\OPT\ge\Cost(X,S)$,
\begin{align*}
\frac{\Cost(x,C)}{\Cost(X,C)}&=\frac{4\gamma\Cost(x,C)}{4\gamma\Cost(X,C)}=\frac{4\gamma\Cost(x,C)}{2\gamma\Cost(X,C)+2\gamma\Cost(X,C)}\le\frac{4\gamma\Cost(x,C)}{2\gamma\Cost(X,C)+2\Cost(X,S)}
\end{align*}
Note that $\Cost(X,S)=\Cost(X,\pi(X))$, so that
\begin{align*}
\frac{\Cost(x,C)}{\Cost(X,C)}&\le\frac{4\gamma\Cost(x,C)}{2\gamma\Cost(X,C)+2\Cost(X,\pi(X))}\\
&\le\frac{4\gamma\Cost(x,C)}{\Cost(X,C)+\Cost(X,\pi(X))+\Cost(X,\pi(X))},
\end{align*}
since $\gamma\ge 1$. 
By the generalized triangle inequality, i.e., \factref{fact:triangle}, we then have $(\Cost(X,C)+\Cost(X,\pi(X))\ge\frac{1}{2^z}\Cost(C,\pi(X))$, so that
\begin{align*}
\frac{\Cost(x,C)}{\Cost(X,C)}&\le\frac{4\gamma\Cost(x,C)}{\frac{1}{2^z}\Cost(C,\pi(X))+\Cost(X,\pi(X))}\\
&\le\frac{2^z\cdot4\gamma\Cost(x,C)}{\Cost(C,\pi(X))+\Cost(X,\pi(X))}\le\frac{2^z\cdot4\gamma\Cost(x,C)}{\Cost(X,C)}.
\end{align*}
Thus, for $\Cost(X,S)=\Cost(X,\pi(X)$ and $Z\le\Cost(X,S)\le\beta Z$, we have
\begin{align*}
\frac{\Cost(x,C)}{\Cost(X,C)}&\le\frac{2^z\cdot4\gamma\Cost(x,C)}{\Cost(C,\pi(X))+Z}\le\frac{2^z\cdot4\beta\gamma\Cost(x,C)}{\Cost(X,C)}.
\end{align*}
Finally, since $Z'\le\Cost(C,\pi(X))\le\beta'Z'$, then we have
\[\frac{\Cost(x,C)}{\Cost(X,C)}\le\frac{2^z\cdot 4\gamma\Cost(x,C)}{Z'+Z}\le\frac{2^z\cdot4\beta\beta'\gamma\Cost(x,C)}{\Cost(X,C)}.\]
\end{proof}


We next bound the sum of the online sensitivities for $(k,z)$-clustering. 
\thmtotalsens*
\begin{proof}
For each $t\in[n]$, we define $X_t=x_1,\ldots,x_t$ and $\OPT(X_t)$ to be the objective of an optimal $(k,z)$ clustering of $X_t$. 
Let $t_0=1,t_1,\ldots,t_y$ be a subsequence of $1,\ldots,n$ such that $\OPT(X_{t_i})>2\OPT(X_{t_{i-1}})$ and either $t_i-1=t_{i-1}$ or $\OPT(X_{t_i-1})\le 2\OPT(X_{t_{i-1}})$ for all $i\in[y]$. 
Note that $y=\O{\log(nd\Delta)}$ for constant $z$ since each of the $n$ points can contribute at most $d\cdot\Delta^z$ cost. 
Hence the sum of the online sensitivities at times $t_i$ where $t_i-1=t_{i-1}$ is at most $\O{\log(nd\Delta)}$. 
We consider the sum of online sensitivities at the remaining times. 

Let $i\in[y]$ be fixed and consider a set $K_{t_i}$ of $k$ centers such that $\cost(X_{t_i},K_{t_i})=\OPT(X_{t_i})$. 
Let $\pi:X_i\to K_i$ be a mapping from each point of $X_{t_i}$ to its closest center in $K_{t_i}$ so that $\pi(x)=\argmin_{p\in K_{t_i}}\dist(x,p)$ for all $x\in X_{t_i}$. 
By the generalized triangle inequality, i.e., \factref{fact:triangle}, we have that for $t\in(t_{i-1},t_i]$ and for any set of $C$ centers, i.e., $C\subseteq[\Delta]^d$ with $|C|=k$,
\[\frac{\dist^z(x_t,C)}{\sum_{j=1}^t\dist^z(x_j,C)}\le\frac{2^{z-1}\dist^z(x_t,\pi(x_t))}{\sum_{j=1}^t\dist^z(x_j,C)}+\frac{2^{z-1}\dist^z(\pi(x_t),C)}{\sum_{j=1}^t\dist^z(x_j,C)}.\]

We upper bound the first term of the right-hand side by lower bounding the denominator. 
Due to the optimality of $K_{t_i}$, we have
\[\sum_{j=1}^t\dist^z(x_j,C)\ge\sum_{j=1}^{t_{i-1}}\dist^z(x_j,C)\ge\sum_{j=1}^{t_{i-1}}\dist^z(x_j,K_{t_{i-1}})=\OPT(X_{t_{i-1}})>\frac{1}{2}\OPT(X_{t_i}).\]
Since $t\in(t_{i-1},t_i]$ and $\OPT(X_{t_i})=\sum_{j=1}^{t_i}\dist^z(x_j,\pi(x_j))\ge\sum_{j=1}^{t}\dist^z(x_j,\pi(x_j))$, then
\[\frac{2^{z-1}\dist^z(x_t,\pi(x_t))}{\sum_{j=1}^t\dist^z(x_j,C)}\le\frac{2^z\dist^z(x_t,\pi(x_t))}{\sum_{j=1}^t\dist^z(x_j,\pi(x_j))}.\]

To bound the second term of the right-hand side, let $S_t$ denote the subset of $X_t$ mapped to $\pi(x_t)$ at time $t_i$, i.e., $S_t=\pi^{-1}(\pi(x_t))\cap X_t$. 
Then 
\[\dist^z(\pi(x_t),C)\cdot|S_t|=\sum_{p\in S_t}\dist^z(\pi(x_t),C).\]
By the generalized triangle inequality, i.e., \factref{fact:triangle},
\begin{align*}
\dist^z(\pi(x_t),C)\cdot|S_t|&=\sum_{p\in S_t}\dist^z(\pi(x_t),C)\\
&\le2^{z-1}\sum_{p\in S_t}\left[\dist^z(\pi(p),p)+\dist^z(p,C)\right]\\
&\le 2^{z-1}\sum_{p\in X_t}\left[\dist^z(\pi(p),p)+\dist^z(p,C)\right],
\end{align*}
since $S_t=\pi^{-1}(\pi(x_t))\cap X_t$ is a subset of $X_t$. 
Because $X_t\subseteq X_{t_i}$, then 
\[\sum_{p\in X_t}\dist^z(\pi(p),p)\le\sum_{p\in X_{t_i}}\dist^z(\pi(p),p)=\OPT(X_{t_i})\le 2\OPT(X_t).\]
By the optimality of $\OPT(X_t)$, 
\[\sum_{p\in X_t}\dist^z(\pi(p),p)\le2\sum_{p\in X_t}\dist^z(p,C)=2\sum_{j=1}^t\dist^z(x_j,C).\]
Hence, 
\[\dist^z(\pi(x_t),C)\cdot|S_t|\le 3\cdot 2^{z-1}\sum_{j=1}^t\dist^z(x_j,C),\]
so that
\[\frac{\dist^z(\pi(x_t),C)}{\sum_{j=1}^t\dist^z(x_j,C)}\le\frac{3\cdot 2^{z-1}}{|S_t|}.\]
\noindent
Putting things together,
\begin{align*}
\frac{\dist^z(x_t,C)}{\sum_{j=1}^t\dist^z(x_j,C)}&\le\frac{2^{z-1}\dist^z(x_t,\pi(x_t))}{\sum_{j=1}^t\dist^z(x_j,C)}+\frac{2^{z-1}\dist^z(\pi(x_t),C)}{\sum_{j=1}^t\dist^z(x_j,C)}\\
&\le\frac{2^z\dist^z(x_t,\pi(x_t))}{\sum_{j=1}^t\dist^z(x_j,\pi(x_j))}+\frac{3\cdot 2^{2z-1}}{|S_t|}.
\end{align*}
\noindent
Therefore,
\begin{align*}
\sum_{t=t_{i-1}+1}^{t_i} \sigma_t&\le\sum_{t=t_{i-1}+1}^{t_i}\left(\frac{2^z\dist^z(x_t,\pi(x_t))}{\sum_{j=1}^t\dist^z(x_j,\pi(x_j))}+\frac{3\cdot 2^{2z-1}}{|S_t|}\right)\\
&=\sum_{t=t_{i-1}+1}^{t_i}\frac{2^z\dist^z(x_t,\pi(x_t))}{\sum_{j=1}^t\dist^z(x_j,\pi(x_j))}+\sum_{t=t_{i-1}+1}^{t_i}\frac{3\cdot 2^{2z-1}}{|S_t|}\\
&\le 2^z+\sum_{t=t_{i-1}+1}^{t_i}\frac{3\cdot 2^{2z-1}}{|S_t|}.
\end{align*}
Since $S_t=\pi^{-1}(\pi(x_t))\cap X_t$ and $\pi^{-1}$ can map onto one of $k$ different sets of points, then 
\[\sum_{t=t_{i-1}+1}^{t_i}\frac{1}{|S_t|}\le k\sum_{t=1}^n\frac{1}{t}.\]
Hence,
\begin{align*}
\sum_{t=t_{i-1}+1}^{t_i} \sigma_t&\le 2^z+(3k\cdot 2^{2z-1})\sum_{t=1}^n\frac{1}{t}=\O{2^{2z}k\log n}.
\end{align*}
Since the sequence $t_0,\ldots,t_y$ satisfies $y=\O{\log(nd\Delta)}$, then 
\[\sum_{t=1}^n \sigma_t=\sum_{i=1}^y\sum_{t=t_{i-1}+1}^{t_i} \sigma_t=\O{2^{2z}k\log^2(nd\Delta)}.\]
\end{proof}

While the input points are from $[\Delta]^d$, the optimal cluster centers may not lie in $[\Delta]^d$. 
Although there is an arbitrary number of subsets of $\mathbb{R}^d$ of size $k$, we now recall the well-known fact that to achieve an approximately optimal clustering, it suffices to only show correctness on a sufficiently-sized net. 
For completeness, we include the proof here. 
\begin{lemma}
\lemlab{lem:kmedian:net}
Let $X\subset[\Delta]^d$ and let $z\ge 1$ be a constant. 
Then there exists a set $S$ of size $|S|=\left(\frac{n\Delta}{\eps}\right)^{\O{kd}}$, such that $(1-\eps)\Cost(A,C)\le\Cost(X,C)\le(1+\eps)\Cost(A,C)$ for any $C\in S$, implies $(1-\eps)\Cost(A,C)\le\Cost(X,C)\le(1+\eps)\Cost(A,C)$ for any set $C\subset\mathbb{R}^d$ with $|C|=k$.
\end{lemma}
\begin{proof}
Note that since $X\subset[\Delta]^d$, then either the optimal $(k,z)$-clustering is zero or at least $\Omega\left(\frac{1}{2^z}\right)$. 
Moreover, note that moving a point or a center by a distance $D$ can only change the clustering cost by $n(\Delta d)^z$. 
Hence, it suffices to create a $D$-net of $[\Delta]^d$ for $D=\O{\frac{\eps}{n(\Delta d)^z}}$, which has size $\left(\frac{n\Delta}{\eps}\right)^{\O{d}}$ for constant $z\ge 1$. 
The claim then follows from the observation that it suffices to choose $k$ points from this net for any set of $k$ queries.  
\end{proof}

We now show the correctness of online sensitivity sampling for $(k,z)$-clustering. 
\begin{lemma}
\lemlab{thm:online:sens:correctness}
Let $\calP$ be a process that for each $t\in[n]$, samples each point $x_t$ with probability $p_t\ge\min(1,\gamma\cdot\sigma_t)$, where $\sigma_t$ is the online sensitivity of $x_t$ and $\gamma=\O{\frac{dk}{\eps^2}\log\frac{n\Delta}{\eps}}$ and gives the point weight $\frac{1}{p_t}$ if $x_t$ is sampled. 
Then with high probability, i.e., $1-\frac{1}{\poly(n)}$, $\calP$ is a $(1+\eps)$-coreset for $(k,z)$-clustering. 
\end{lemma}
\begin{proof}
Let $C$ be a set of $k$ centers. 
Let $x_1,\ldots,x_n\in\mathbb{R}^d$ be a sequence of points and for $t\in[n]$, let $X_t=\{x_1,\ldots,x_t\}$. 
Let $Y_0=0,Y_1,\ldots,Y_n$ be a martingale with difference sequence $Z_0,\ldots,Z_n$, where for $t\ge 1$, we define $Z_t=0$ if $|Y_{t-1}|>\eps\Cost(C,X_{t-1})$. 
Otherwise, for $|Y_{t-1}|\le\eps\Cost(C,X_{t-1})$, we define
\[X_t=\begin{cases}\left(\frac{1}{p_t}-1\right)\Cost(C,x_t)\qquad&\text{if }x_t\text{ is sampled}\\
-\Cost(C,x_t)\qquad&\text{otherwise}
\end{cases}.\]
Note that $\Ex{X_t}=0$ and thus $\Ex{Y_t\mid Y_1,\ldots,Y_{t-1}}=Y_{t-1}$, so that $Y_0,\ldots,Y_t$ is a valid martingale sequence. 
Furthermore, by the construction of $Y_t$, we have that $Y_t=\Cost(C,A_t)-\Cost(C,X_t)$, where $A_t$ is the coreset at time $t$. 

Observe that if $p_t=1$, then $X_t=0$. 
Otherwise, we have that $\Ex{X_t^2}\le\frac{1}{p_t}\Cost(C,x_t)^2$. 
Let $\gamma=\frac{100}{\eps^2}\sigma_t\log\frac{1}{\delta}$ for some parameter $\delta$ that we will set, which will ultimately give the claimed setting for $\gamma$. 
Since $p_t\ge\gamma\cdot\sigma_t$, then $p_t\ge\frac{100\log\frac{1}{\delta}}{\eps^2}\sigma_t$, then $p_t\ge\frac{100\log\frac{1}{\delta}}{\eps^2}\frac{\Cost(C,x_t)}{\Cost(C,X_t)}$, so that $\Ex{X_t^2}\le\frac{\eps^2}{100\log\frac{1}{\delta}}\Cost(C,x_t)\cdot\Cost(C,X_t)$. 
Hence, 
\[\sum_{i=1}^t\Ex{X_i^2}\le\sum_{i=1}^t\frac{\eps^2}{100}\Cost(C,x_i)\cdot\Cost(C,X_i)\le\frac{\eps^2}{100\log\frac{1}{\delta}}\Cost(C,X_t)^2.\]
Furthermore, we have $X_i\le\frac{\eps^2}{100\log\frac{1}{\delta}}\Cost(C,X_i)$ for all $i\in[t]$. 
Then by Freedman's inequality, i.e., \thmref{thm:scalar:freedman},
\[\PPr{\max_{i\in[t]}|Y_i|\ge\eps\Cost(C,X_t)}\le2\exp\left(-\frac{\eps\Cost(C,X_t)^2/2}{\frac{\eps^2}{100\log\frac{1}{\delta}}\Cost(C,X_t)^2+\frac{\eps^3}{300\log\frac{1}{\delta}}\Cost(C,X_t)^2}\right),\]
and so in particular,
\[\PPr{|\Cost(C,A_t)-\Cost(C,X_t)|\le\eps\Cost(C,X_t)}\ge1-\delta.\]
By \lemref{lem:kmedian:net} for each time $t\in[n]$, it suffices to union bound over a net of size $\left(\frac{n\Delta}{\eps}\right)^{\O{kd}}$ to ensure that 
\[|\Cost(C,A_t)-\Cost(C,X_t)|\le2\eps\Cost(C,X_t),\]
for any set $C\subset\mathbb{R}^d$ with $|C|=k$. 
We further union bound over all $t\in[n]$ to ensure that 
\[|\Cost(C,A_t)-\Cost(C,X_t)|\le2\eps\Cost(C,X_t),\]
for any set $C\subset\mathbb{R}^d$ with $|C|=k$ and for any $t\in[n]$. 
Thus by setting $\log\frac{1}{\delta}=\O{kd\log\frac{n\Delta}{\eps}}$, we have that with high probability, $\calP$ is a coreset for $(k,z)$-clustering at all times in the stream. 
\end{proof}

\begin{lemma}
\lemlab{thm:online:sens:space}
Let $\gamma=\O{\frac{dk}{\eps^2}\log\frac{n\Delta}{\eps}}$ and $\tau>1$ be some parameter. 
Let $\calP$ be a process that for each $t\in[n]$, samples each point $x_t$ with probability $p_t$, where 
\[p_t\le\min(1,\tau\gamma\sigma_t),\]
for the online sensitivity $\sigma_t$ of $x_t$. 
Then with high probability, $\calP$ contains at most $\O{\frac{\tau dk^2}{\eps^2}\log\frac{n\Delta}{\eps}}$ points. 
\end{lemma}
\begin{proof}
For each $t\in[n]$, let $Y_t$ denote the indicator random variable for whether $x_t$ is sampled by $\calP$, i.e., $Y_t=1$ if $x_t$ is sampled by $\calP$ and $Y_t=0$ otherwise, so that $Y:=\sum_{t=1}^n Y_t$ is the total number of sampled points. 
We have that $\Ex{Y_t}=p_t$ for each $t\in[n]$ and thus 
\[\Ex{Y}=\sum_{t=1}^n\Ex{Y_t}=\sum_{t=1}^n p_t\le \tau\gamma\cdot\sum_{t=1}^n\sigma_t.\]
By \thmref{thm:total:online:sens}, we have
\[\sum_{t=1}^n\sigma_t=\O{2^{2z}k\log^2(nd\Delta)},\]
so that 
\[\Ex{Y}\le\tau\gamma\cdot\O{2^{2z}k\log^2(nd\Delta)}.\]
Since $p_t\in[0,1]$, we have $\Ex{Y_t^2}\le p_t$ and so $\sum_{t=1}^n\Ex{Y_t^2}\le\tau\gamma\cdot\sum_{t=1}^n\sigma_t$. 
Hence by Freedman's inequality, i.e., \thmref{thm:scalar:freedman}, we have that for $\gamma=\O{\frac{dk}{\eps^2}\log\frac{n\Delta}{\eps}}$, $Y=\O{\frac{\tau dk^2}{\eps^2}\log\frac{n\Delta}{\eps}}$ with high probability. 
\end{proof}

\begin{theorem}[Online sensitivity sampling]
\thmlab{thm:online:sens}
Let $\calP$ be a process that for each $t\in[n]$, samples each point $x_t$ with probability $p_t\ge\min(1,\gamma\cdot\sigma_t)$, where $\sigma_t$ is the online sensitivity of $x_t$ and $\gamma=\O{\frac{dk}{\eps^2}\log\frac{n\Delta}{\eps}}$ and gives the point weight $\frac{1}{p_t}$ if $x_t$ is sampled. 
Then with high probability, $\calP$ is a coreset for $(k,z)$-clustering that contains at most $\O{\frac{dk^2}{\eps^2}\log\frac{n\Delta}{\eps}}$ points. 
\end{theorem}
\begin{proof}
We argue by induction at each time $t\in[n]$. 
At time $t=1$, the online sensitivity of the first point is $1$ and thus the first point is sampled, and so $\calP$ is a coreset of $x_1$. 

Now suppose that at some time $t\in[n]$, $\calP$ is a coreset of $X_t=\{x_1,\ldots,x_t\}$. 
Then by the definition of coreset, the cost of any set $C$ of $k$ centers will be preserved by $\calP$. 
Hence, the coreset will output a $(1+\eps)$-approximation $q'_{t+1}$ to the online sensitivity of the next point $x_{x+1}$. 
Then by setting $q_{t+1}=2q'_{t+1}$, we have that $2\sigma_{t+1}>q_{t+1}>\sigma_{t+1}$ and so by setting $p_{t+1}=\gamma\cdot q_{t+1}$ for $\gamma=\O{\frac{dk}{\eps^2}\log\frac{n}{\eps}}$, the conditions of \lemref{thm:online:sens:correctness} are satisfied, so that with high probability, $\calP$ remains a coreset at time $t+1$. 
It follows by induction that $\calP$ is a coreset at all times. 

Moreover, it follows that each point $x_t$ at time $t\in[n]$ is sampled with probability at most $2\gamma\sigma_t$. 
Hence by \lemref{thm:online:sens:space}, we have that with high probability, $\calP$ contains $\O{\frac{dk^2}{\eps^2}\log\frac{n\Delta}{\eps}}$ points. 
\end{proof}

\subsection{\texorpdfstring{$(k,z)$}{(k,z)}-Clustering on Insertion-Only Streams}
\seclab{sec:kz:insertion}
In this section, we present a simple algorithm for $(k,z)$-clustering on insertion-only streams using $o(\log n)$ words of space. 
The algorithm proceeds as follows. 
At all times, we maintain a data structure that with high probability, produces a $(1+\eps)$-coreset $C_t$ to the point sets $X_t=\{x_1,\ldots,x_t\}$ at all times $t\in[n]$. 
Conditioned on the correctness of the coreset $C_{t-1}$ at time $t-1$, we use $C_{t-1}$ to approximate the online sensitivity of point $x_t$, which in turn, produces a probability $p_t$ of sampling $x_t$ into a secondary stream $\calS'$. 
We do not maintain $\calS'$, but rather feed it into a coreset construction algorithm, which then produces $C_t$, completing the process for time $t$. 

\thmmain*
\begin{proof}
For $t\in[n]$, let $\calE_t$ be the event that the data structure $C_t$ at time $t$ is a $(1+\eps)$-coreset of the data set $X_t=\{x_1,\ldots,x_t\}$. 
For all $t\in[n]$, let $\calS'_t$ be the stream consisting of the weighted points sampled by sensitivity sampling. 
Observe that conditioned on $\calE_{t-1}$, we can compute a $(1+\eps)$-approximation $q_t$ to the online sensitivity $\sigma_t$ of the point $x_t$. 
Thus by setting $p_t=\min(1,2\gamma q_t)$ for $\gamma=\O{\frac{dk}{\eps^2}\log\frac{n\Delta}{\eps}}$, then by \lemref{thm:online:sens:correctness}, we have that $\calS'_t$ is a $(1+\eps)$-coreset of $X_t$. 
Let $\calF_t$ be the event that $\calS'_t$ is a $(1+\eps)$-coreset of $X_t$. 
By applying a coreset construction with failure probability $\frac{1}{\poly(n)}$ and a rescaling of $\eps$, it follows that $C_t$ is a $(1+\eps)$-coreset of $X_t$, conditioned on $\calE_{t-1}$ and $\calF_t$. 
By \lemref{thm:online:sens:space}, we have that $|\calS'_t|\le\poly\left(d,k,\frac{1}{\eps},\log n,\log\Delta\right)$. 
Hence, we have that $\sum_{t=1}^n \PPr{\not\calF_t\,\mid\,\calE_{t-1}}\le\frac{1}{\poly(n)}$ and $\sum_{t=1}^n \PPr{\not\calE_t\,\mid\,\calF_t,\calE_{t-1}}\le\frac{1}{\poly\left(d,k,\frac{1}{\eps},\log n,\log\Delta\right)}$.
Thus by a union bound, we have that with high probability, $C_t$ is a $(1+\eps)$-coreset for $X_t$ at all times $t\in[n]$. 
Finally by \thmref{thm:offline:coreset:size}, we have that our algorithm uses $\tO{\frac{dk}{\varepsilon^2}}\cdot(2^{z\log z})\cdot\min\left(\frac{1}{\varepsilon^z},k\right)\cdot\poly(\log\log(n\Delta))$ words of memory. 
\end{proof}

\section{Two-Pass Dynamic Streaming Algorithm for \texorpdfstring{$k$}{k}-Median Clustering}
\seclab{sec:twopass:kmed}
In this section, we present a two-pass algorithm that uses $o(\log n)$ words of space for $k$-median clustering, using the techniques that we have built up over the previous sections. 
Our algorithm outputs a weighted subset of the input points, which forms a $(1+\eps)$-coreset of the underlying point set. 
In \secref{sec:lb}, we show it is impossible to achieve this in $o(\log n)$ words of space using a single pass over the data. 

We first require preliminaries on $L_1$ norm estimation in an insertion-only stream and on the notion of consistent clustering. 
We then incorporate these ideas into our analysis along with structural results regarding sensitivity sampling similar to those that were given in the previous section. 

\subsection{\texorpdfstring{$L_1$}{L1} Sketch}
We start by describing a streaming algorithm for $L_1$-approximation. 
In this setting, the input is a stream of elements $u_1,\ldots,u_m\in[n]$ that implicitly defines an underlying frequency vector $f\in\mathbb{R}^n$, so that each update $u_t=(c_t,\Delta_t)$ changes coordinate $c_t\in[n]$ by $\Delta_t\in[-M,M]$ for $M=\poly(n)$. 
More formally, for all $i\in[n]$, 
\[f_i=|\sum_{t\in[m]: c_t=i}\Delta_t.\]
The problem is then to output a $(1+\eps)$-approximation to 
\[\|f\|_1=\sum_{i=1}^n|f_i|=|f_1|+\ldots+|f_n|.\]

We define the general family of $p$-stable distributions, though we shall ultimately only require Cauchy random variables, i.e., the $p$-stable distribution with $p=1$. 

\begin{definition}[$p$-stable distribution]
\cite{Zolotarev89}
\deflab{def:pstable}
For $0<p\le 2$, there exists a distribution $\calD_p$ called the $p$-stable distribution $\calD_p$, such that for any positive integer $n$ with $Z_1,\ldots,Z_n\sim\calD_p$ and vector $x\in\mathbb{R}^n$, then $\sum_{i=1}^n Z_ix_i\sim Z\cdot\|x\|_p$ for $Z\sim\calD_p$. 
The $p$-stable distribution has probability density function $f(x)=\Theta\left(\frac{1}{1+|x|^{1+p}}\right)$ for $p\le 2$ and is equal to the Cauchy distribution for $p=1$ (and also the normal distribution for $p=2$).  
\end{definition}
For completeness, \cite{Nolan03} presented a now standardized method for sampling a random variable $X$ from the $p$-stable distribution, by generating $\theta$ uniformly at random from the interval $\left[-\frac{\pi}{2},\frac{\pi}{2}\right]$, $r$ uniformly at random from the interval $[0,1]$, and outputting
\[X=f(r,\theta)=\frac{\sin(p\theta)}{\cos^{1/p}(\theta)}\cdot\left(\frac{\cos(\theta(1-p))}{\log\frac{1}{r}}\right)^{\frac{1}{p}-1}.\]

The $p$-stable distribution is utilized in an $L_1$ estimation algorithm of \cite{Indyk06}. 
The algorithm creates $\ell=\O{\frac{1}{\eps^2}\left(\log\frac{1}{\eps\delta}+\log\log m\right)}$ vectors $v^{(1)},\ldots,v^{(\ell)}\in\mathbb{R}^n$. 
For each $i\in[\ell]$, the vector $v^{(i)}$ consists of independently generated Cauchy random variables and the algorithm maintains $Z_i=\langle v^{(i)},f\rangle$ throughout the stream. 
At the end of the stream, the algorithm outputs $\median_{i\in[\ell]} Z_i$ as the estimation for $\|f\|_1$. 

We have the following guarantees for the algorithm of $p$-stable sketch:
\begin{theorem}
\thmlab{thm:lp:strong:tracking:bdn}
\cite{Indyk06}
For $p\in(0,2]$ and any dynamic stream of length $m$ on a universe of size $n$, there exists an algorithm that with probability at least $1-\delta$, provides a $(1+\eps)$-approximation to the $L_p$ norm of the stream, using $\O{\frac{1}{\eps^2}\log nm\log\frac{1}{\delta}}$ bits of space. 
\end{theorem}
Putting the $L_1$ sketch and the EMD embedding together, we have:
\begin{lemma}
\lemlab{lem:sens:lower}
There exists an algorithm that uses $\O{kd(\log k+\log\log\Delta)}$ words of space and outputs a clustering with estimated cost $Z$ such that with probability at least $0.98$, 
\[\OPT\le Z\le\O{d^{1.5}(\log k+\log\log\Delta)}\OPT,\]
where $\OPT$ is the optimal $k$-median clustering cost at the end of the stream. 
\end{lemma}

\subsection{Dynamic Streaming Algorithm}
We now describe our two pass streaming algorithm for insertion-deletion streams. 
In the first pass over the data stream, our algorithm maintains a sketch of the EMD embedding, as in \algref{alg:emd:sketch}. 

\begin{algorithm}[!htb]
\caption{EMD Sketch}
\alglab{alg:emd:sketch}
\begin{algorithmic}[1]
\Require{Input points $x_1,\ldots,x_n\in[\Delta]^d$}
\Ensure{Data structure to support EMD queries}
\State{Initiate random shift $s$}
\State{Create grid embedding $G_s$}
\State{Let $\mu_t$ be a probability distribution and $\nu$ be a $k$-sparse probability distribution}
\State{\Return a constant-factor approximation to $\|G_s(\mu_t-\nu)\|_1$ using an $L_1$ sketch, i.e., \thmref{thm:lp:strong:tracking:bdn}}
\end{algorithmic}
\end{algorithm}

Because the data structure is a linear sketch, it is amenable to insertions and deletions. 
We would like to use the EMD sketch to perform sensitivity sampling over the points, where we define the sensitivity of a point $x\in\mathbb{R}^d$ with respect to a set $X\subset\mathbb{R}^d$ by
\[\varphi_X(x):=\max_{C\subset\mathbb{R}^d: |C|\le k}\frac{\Cost(x,C)}{\Cost(X,C)}=\max_{C\subset\mathbb{R}^d: |C|\le k}\frac{\dist(x,C)^z}{\sum_{t=1}^n\dist(x_t,C)^z}.\]
Throughout the section, the set $X$ will be the set of points at the end of the stream and thus we omit the subscript $X$ for clarity. 
We show that that an approximation to the sensitivity of each point can be computed in \algref{alg:approx:sens}. 

\begin{algorithm}[!htb]
\caption{Approximation of Sensitivity for $k$-Median Clustering}
\alglab{alg:approx:sens}
\begin{algorithmic}[1]
\Require{Input dataset $X\subset[\Delta]^d$}
\Ensure{Approximate sensitivity for $x$ for any query $x\in[\Delta]^d$}
\State{$m\gets\tO{kd\log\log\Delta}$}
\State{Initialize an EMD sketch, i.e., $m$ instances of $L_1$ sketches on random EMD embeddings}
\State{Update the EMD sketch with $X$}
\Comment{First pass}
\State{Use the EMD sketch to find a near-optimal set $S$ of $k$ centers}
\State{Let $\widetilde{\Cost(X,S)}$ be the estimated cost of $S$ on $X$}
\Comment{\lemref{lem:sens:lower}}
\State{Let $\calM$ be the net from \lemref{lem:kmedian:net}}
\State{$q(x)\gets0$}
\For{$C\in\calM$}
\State{$q(x)\gets\max\left(q(x),\frac{8\gamma\Cost(x,C)}{\Cost(S,C)+\widetilde{\Cost(X,S)}}\right)$}
\Comment{\lemref{lem:opt:approx:sens}}
\EndFor
\State{\Return $q(x)$}
\end{algorithmic}
\end{algorithm}

Unfortunately, sensitivity sampling is incompatible with the dynamic setting, where sampled points could be subsequently deleted. 
Instead, in the second pass over the data stream, our algorithm uses the EMD sketch to subsample from the universe with probability proportional to each point's sensitivity.  
Thus if a universe element is subsampled, all updates to the element throughout the stream are reported. 
The universe element is then fed into a sparse recovery scheme, which is a linear sketch that can handle both insertions and deletions. 
We present the algorithm in full in \algref{alg:two:pass:dynamic:kmed}. 

\begin{algorithm}[!htb]
\caption{Two-pass Dynamic Streaming Algorithm for $k$-Median}
\alglab{alg:two:pass:dynamic:kmed}
\begin{algorithmic}[1]
\Require{Stream of updates of length $m=\poly(n)$ to coordinates of $[\Delta]^d$, approximation parameter $\eps\in(0,1)$, number of clusters $k$, parameter $z\ge 1$}
\Ensure{$(1+\eps)$-coreset for $(k,z)$-median clustering}
\State{First pass: implicitly compute an approximate sensitivity $q(x)$ for all $x\in[\Delta]^d$}
\Comment{\algref{alg:approx:sens}}
\State{$s\gets\poly\left(k,d,\frac{1}{\eps^2},\log\log n,\log\log\Delta\right)$}
\State{$p(x)=\min\left(1,\frac{k^2d}{\eps^2}\log k\cdot q(x)\right)$}
\State{Initialize an $100s$-sparse recovery algorithm $A$ (see \lemref{lem:sparse:recovery})}
\State{Draw hash function $h$ from a family of hash functions where for all $x\in\mathbb{R}^d$, $\PPr{h(x)=1}=p(x)$ and $h(x)=0$ otherwise}
\For{each update in the second pass}
\If{the update is to $x\in[\Delta]^d$ with $h(x)=1$}
\State{Update $A$ with the corresponding update to $x$}
\EndIf
\EndFor
\State{\Return the points $x$ output by $A$ weighted by $\frac{1}{q(x)}$}
\end{algorithmic}
\end{algorithm}

We first claim that with constant probability, we can obtain a bounded approximation to the sensitivity of each point. 
\begin{lemma}
\lemlab{lem:sens:approx}
There exists an algorithm that uses $\O{kd^2\log\log\Delta(\log k+\log\log n)}$ words of space and outputs $q(x)$ for all points $x\in\mathbb{R}^d$ in the stream, such that with probability at least $0.98$, 
\[\varphi(x)\le q(x)\le\O{d^{1.5}(\log k+\log\log\Delta)}\cdot\varphi(x),\]
where $\varphi(x)$ is the sensitivity of $x$ with respect to a point set $X\subset\mathbb{R}^d$ for the $k$-median problem.  
\end{lemma}
\begin{proof}
Let $\widetilde{Z}$ be the smallest estimated cost by the EMD sketch scaled by $\O{\sqrt{d}}$, across all sets of $k$ centers in the doubling dimension net. 
Let $\OPT$ be the optimal $k$-median clustering cost of the dataset at the end of the stream. 
Let $\calE$ be the event that 
\[\OPT\le\widetilde{Z}\le\O{d^{1.5}(\log k+\log\log\Delta)}\OPT,\]
so that $\PPr{\calE}$ with probability at least $0.98$, by \lemref{lem:sens:lower}. 

Recall that for a point $x$, its sensitivity is defined by 
\[\varphi(x)=\max_{C:\,C\subset\mathbb{R}^d,\, |C|=k}\frac{\Cost(x,C)}{\Cost(X,C)},\]
where $X$ is the set of the points remaining at the end of the stream.

By \lemref{lem:opt:approx:sens}, a constant factor approximation to a clustering that achieves a clustering cost that is an $\O{d^{1.5}(\log k+\log\log\Delta)}$-approximation to $\OPT$ can be used to compute a $\O{d^{1.5}(\log k+\log\log\Delta)}$-approximation to $\frac{\Cost(x,C)}{\Cost(X,C)}$ for any set $C$ of $k$ centers. 
Thus conditioned on $\calE$ the estimate $\widetilde{Z}$ of the EMD sketch can be used to compute a $\O{d^{1.5}(\log k+\log\log\Delta)}$-approximation to $\varphi(x)$ for all $x\in[\Delta]^d$. 
Hence, the EMD sketch can only be used to compute $q(x)$ such that with probability at least $0.98$, 
\[\varphi(x)\le q(x)\le\O{d^{1.5}(\log k+\log\log\Delta)}\cdot\varphi(x),\]
simultaneously for all $t\in[n]$. 

The space complexity results from maintaining $\tO{kd\log\log\Delta}$ instances of the EMD sketch, i.e., the $L_1$ sketch on the EMD embedding. 
Note that since the grid of the EMD embedding has size $\Delta^{\O{d}}$, then each $L_1$ sketch uses $\O{d\log\Delta}$ bits of space. 
Hence, the total space is $\tO{kd^2\log\log\Delta}$ words of space. 
\end{proof}

We then recall a well-known deterministic algorithm for sparse recovery.
\begin{lemma}
\lemlab{lem:sparse:recovery}
There exists a deterministic algorithm that recovers the non-zero coordinates of a $k$-sparse vector of length $n$ in a dynamic stream, using $\O{k\log n}$ bits of space. 
\end{lemma}

We now show that if we perform sensitivity sampling on the universe, then the number of sampled elements in the stream is upper bounded by $100s$, so that the sparse recovery scheme used in \algref{alg:two:pass:dynamic:kmed} will properly recover all sampled points. 
\begin{lemma}
\lemlab{lem:sample:sparse}
Let $\calE$ be the event that for all $x\in X$, we have
\[\varphi(x)\le q(x)\le\O{d^{1.5}(\log k+\log\log\Delta)}\cdot\varphi(x).\]
Then conditioned on $\calE$, the set of nonzero points induced by the input to $A$ has sparsity at most $100s$ with probability at least $0.99$. 
\end{lemma}
\begin{proof}
Although as many as $n$ points may be input to $A$, the nonzero points at the end of the stream form exactly $X$. 
Because $A$ is a sparse recovery algorithm, it suffices to consider the points $x$ of $X$ with $h(x)=1$. 
Observe that conditioned on $\calE$, the expected number of sampled points is $\sum_{x\in X}p(x)\le s$. 
Thus by Markov's inequality, the set of nonzero points induced by the input to $A$ has sparsity at most $100s$ with probability at least $0.99$. 
\end{proof}

We next recall the following statement showing that sensitivity sampling results in a $(1+\eps)$-coreset of the underlying dataset. 
\begin{theorem}[Theorem B.9 in \cite{WoodruffY23a}]
\thmlab{thm:sens:sample}
For $X\subset[\Delta]^d$ of $n$ points, let $\varphi(x)$ be the sensitivity of $x\in X$, let $q(x)\ge\varphi(x)$, and let $p(x)\ge\min\left(1,\frac{k^2d}{\eps^2}\log k\cdot q(x)\right)$. 
Then with probability at least $0.99$, independently sampling each point $x\in X$ with probability $p(x)$ and reweighting each sampled point by $\frac{1}{p(x)}$ gives a $(1+\eps)$-coreset.
\end{theorem}
Importantly, the proof of \thmref{thm:sens:sample} in \cite{WoodruffY23a} relies on a simple Bernstein's inequality, i.e. \thmref{thm:bernstein} and a union bound over a net whose size is determined by the VC dimension. 
Thus by applying the following version of Bernstein's inequality for $t$-wise independent random variables with $t=\tO{\frac{d^2k^2}{\eps^2}}\cdot\poly(\log\log(n\Delta))$, the same result holds. 
\begin{theorem}
\thmlab{thm:hash:derandomize}
\cite{Skorski22}
Let $X_1,\ldots,X_n$ be $t$-wise independent random variables with $|X_i-\Ex{X_i}|\le 1$. 
Let $X=\sum_{i=1}^n X_i$ and $V=\sum_{i=1}^n\Var{X_i}$. 
Then there exists an absolute constant $C>0$ such that for $d\le k$ with $\log\frac{d}{V}<\max\left(\frac{d}{n},2\right)$,
\[\PPr{|X-\Ex{X}|>\gamma}\le C\cdot\left(\frac{\sqrt{dV}}{\gamma}\right)^d.\]
\end{theorem}
In particular, note that we can normalize each independent random variable by the total cost of $X$ with respect to a fixed set $C$ of $k$ centers, so that $|X_i-\Ex{X_i}|\le 1$. 
The main point of \thmref{thm:hash:derandomize} is that we can now derandomize the hash function $h$ by instead using a $t$-wise independent hash function on a universe of size $\poly(n)$, which can be generated and stored using $\O{t\log n}$ bits of space~\cite{WegmanC81}. 

We now justify the correctness of \algref{alg:two:pass:dynamic:kmed}. 
\begin{lemma}
\lemlab{lem:twopass:correctness}
Let $Z$ be the set of weight samples output by $A$ at the end of the stream. 
With probability at least $0.96$, $Z$ is a $(1+\eps)$-coreset for $X$ for $k$-median clustering.
\end{lemma}
\begin{proof}
By \lemref{lem:sens:approx}, we have that with probability at least $0.98$,
\[\varphi(x)\le q(x)\le\O{d^{1.5}(\log k+\log\log\Delta)}\cdot\varphi(x).\]
Let $Z_0$ be the set of weighted points obtained by sampling each $x\in X$ with probability $p(x)$ and reweighting by $\frac{1}{p(x)}$, where $p(x)=\min\left(1,\frac{k^2d}{\eps^2}\log k\cdot q(x)\right)$
By \thmref{thm:sens:sample}, $Z_0$ will be a $(1+\eps)$-coreset of $X$ with probability at least $0.99$. 
By \lemref{lem:sample:sparse}, $Z_0$ will have at most $100m$ points with probability at least $0.99$, in which case these points will be recovered by $A$, by \lemref{lem:sparse:recovery}. 
Thus by a union bound, $Z$ is a $(1+\eps)$-coreset for $X$ for $k$-median clustering with probability at least $0.96$. 
\end{proof}
Putting things together, we obtain the following guarantees of \algref{alg:two:pass:dynamic:kmed}.
\begin{theorem}
Given an accuracy parameter $\eps\in(0,1)$, an integer $k>0$ for the number of clusters, and a dynamic stream of length $m=\poly(n)$ defining a set of points $X\subseteq[\Delta]^d$ with $X=\{x_1,\ldots,x_n\}$, there exists a one-pass streaming algorithm that uses $\tO{\frac{d^{2.5}k^2}{\varepsilon^2}}\cdot\poly(\log\log(n\Delta))$ words of memory and outputs a $(1+\eps)$-approximation to $k$-median clustering. 
\end{theorem}
\begin{proof}
Consider \algref{alg:two:pass:dynamic:kmed}. 
Correctness follows from \lemref{lem:twopass:correctness}. 
To analyze the space complexity, observe that the first pass sets up the EMD sketch, which by \lemref{lem:sens:approx}, uses a total of $\tO{kd^2\log\log\Delta}$ words of space. 

Observe that by \lemref{lem:sample:sparse}, the sparse recovery data structure in the second pass uses space $\tO{\frac{d^{2.5}k^2}{\varepsilon^2}}\cdot\poly(\log\log(n\Delta))$ after derandomization.  
Therefore, the total space used by \algref{alg:two:pass:dynamic:kmed} is $\tO{\frac{d^{2.5}k^2}{\varepsilon^2}}\cdot\poly(\log\log(n\Delta))$. 
\end{proof}

\section{\texorpdfstring{Wasserstein-$z$}{Wasserstein-z} Embedding and \texorpdfstring{$(k,z)$-Clustering}{(k,z)-Clustering}}
\seclab{sec:wassz}
In this section, we present our Wasserstein-$z$ embedding and how to subsequently utilize it for $(k,z)$-clustering on dynamic streams. 
We first present the embedding in \secref{sec:wassz:embed}. 
The algorithm and analysis for $(k,z)$-clustering are presented in \secref{sec:kz:alg}. 

\subsection{\texorpdfstring{Wasserstein-$z$}{Wasserstein-z} Embedding}
\seclab{sec:wassz:embed}
We first describe our adaptation of the standard quadtree embedding technique to achieving a bicriteria approximation for $(k,z)$-clustering. 
As in the EMD sketch, we select a shift $s=(s_1,\ldots,s_d)\in[\Delta]^d$ uniformly at random. 
For the sake of presentation, it will be easier to add $s$ to each of the input points $x_1,\ldots,x_n\in[\Delta]^d$ in this case, rather than shifting the underlying geometric data structure by $s$ as in the case of the EMD sketch. 
We then create a tree in the following manner.

Assume without loss of generality that $\Delta$ is a power of two. 
The root of the tree represents a grid with side length $2\Delta$ on the non-negative orthant, so that the origin lies on a corner of the grid. 
Because $s,x_1,\ldots,x_n\in[\Delta]^d$, then all shifted points are in this grid with side length $2\Delta$. 
We thus say that all points belong to the root node of the tree. 

We then partition the grid into $2^d$ smaller hypercubes with side length $\Delta$. 
For each of these hypercubes that contains a shifted input point, we create a node representing the hypercube and add the node as a child of the root node in the tree, using an edge with weight $\sqrt{d}\Delta$. 
We define these edges as having height zero, connecting the root node and the children nodes at height one. 
Observe that $\sqrt{d}\Delta$ is the maximum distance between any two points in the parent cube, since the original input points satisfy $x_1,\ldots,x_n\in[\Delta]^d$. 
Moreover, since there are only $n$ input points, then the root node has at most $n$ children nodes at height one. 
We repeat this procedure until every node representing a hypercube contains at most a single point, resulting in a tree where all leaves contain a single point and are all at the same height, which is at most $\O{\log\Delta}$. 

As it will be useful to index from the leaves of the tree, we will also define the $t$-th level of the quadtree to be the hypergrid with length $2^t$, as opposed to height $t$ corresponding to the hypergrid with length $\frac{\Delta}{2^{t-1}}$, recalling that we assume without loss of generality $\Delta$ is a power of two. 

We first embed the input points $X$ into the quadtree. 
That is, for $\mu\in\mathbb{R}^{[\Delta]^d}$, we define $W_{s,t}\mu$ as the frequency vector over the hypercubes of the hypergrid $\calG_{s,t}$ at level $t$ that counts the total mass in each hypercube. 
To avoid ambiguity, we again say that each cell of a grid of length $2^t$ has closed boundaries on one side and open boundaries on the other side, e.g., a cell that contains a point $(x_1,\ldots,x_d)$ may contain points $(y_1,\ldots,y_d)$ where $y_i<x_i+2^k$, but does not contain any points $(y_1,\ldots,y_d)$ where $y_i\ge x_i+2^k$. 
As before, we define the mapping $W_s\mu$ to be the concatenation of the vectors 
\[W_s\mu=(W_{s,0}\mu)\circ((2\sqrt{d})^z\cdot W_{s,1}\mu)\circ\ldots\circ((2^z\sqrt{d})^{z}\cdot W_{s,t}\mu)\circ\ldots\circ((2^\ell\sqrt{d})^{z}\cdot W_{s,\ell}\mu).\]

Now, for a query point $q$, we say that $q$ is \emph{bad at level $i$} if there exists a hyperplane of the quadtree decomposition of length $2^i$ that has distance less than $\frac{2^i}{d\log\Delta}$ from $q$. 
Otherwise, we say that $q$ is good at level $i$. 
Observe that if $q$ is good at level $i$, but $q$ and $x\in X$ are first separated at level $i$, then $\|q-x\|_2^z>\frac{2^{iz}}{d^{z}\log^{z}\Delta}$ and the incurred estimated cost for $q$ and $x$ in the quadtree is most $(2^i\sqrt{d})^z$, so the distortion will be at most $d^{0.5z}\log^{z}\Delta$. 

On the other hand, if $q$ is bad at level $i$, then by definition, it has distance less than $\frac{2^i}{d\log\Delta}$ from a hyperplane of the quadtree decomposition of length $2^i$ and so $q$ and $x$ incur estimated cost $2^{iz}d$ by the quadtree, then the distortion could be significantly larger. 

Thus for a query point $q$, we define a mapping $\phi$ to subsets $S_q$ of $\mathbb{R}^d$ as follows. 
That is, we set $\phi(q)=S_q$ and describe a well-defined iterative process to construct $S_q$. 

We initialize $S_q=\{q\}$, i.e., we first add $q$ to $S_q$. 
If $q$ is bad at level $i$, then for each hyperplane $H$ of the quadtree of length $2^i$ that is too close to $q$, i.e., $\dist(q,H)\le\frac{2^i}{d\log\Delta}$, we create a copy $q^{(H)}_i$ corresponding to the projection of $q$ onto $H$ (or the closest point on the projection line on the other side of $H$), so that $q^{(H)}_i$ is responsible for serving the points assigned to $q$ that are on the other side of $H$. 
We add the point $q^{(H)}_i$ to $S_q$ for each hyperplane $H$ for which $\dist(q,H)\le\frac{2^i}{d^2\log\Delta}$, noting that $q$ could induce multiple points $q^{(H)}_i$ to be added to $S_q$. 
We then proceed top-down across the levels of $i$, repeatedly adding points to $S_q$ as necessary.
See \figref{fig:wassz:quadtree} for an example of this process. 

We first show that $\phi(C)$ contains at most $\O{k}$ centers in expectation.

\begin{lemma}
\lemlab{lem:new:centers:project}
Given a set $C=\{C_1,\ldots,C_k\}$ of $k$ centers, $\phi(C)$ contains at most $\O{k}$ centers in expectation over the choice of the randomized shift parameter $s\in[\Delta]^d$. 
\end{lemma}
\begin{proof}
Consider a level $i$ and a center $C_j$ with $j\in[k]$. 
Let $\calH$ be the set of hyperplanes of the hypergrid at level $i$. 
The probability that $\dist(C_j,\calH)<\frac{2^i}{d\log\Delta}$ is at most 
\[\O{\frac{d}{d\log\Delta}}=\O{\frac{1}{\log\Delta}}.\]
Hence, summing up over all the $\O{\log\Delta}$ levels, the expected number of center copies that a center $C_j$ can generate is at most $\O{1}$. 
Finally, summing up over $j\in[k]$, then the expected number of center copies that $C$ can generate, i.e., the number of additional points in $\phi(C)$ besides $C$, is at most $\O{k}$. 
\end{proof}

By construction, the resulting mapping $\phi(C)$ for a set of $k$ centers $C=\{C_1,\ldots,C_k\}$ will induce an estimated cost by the quadtree that has small distortion from the actual cost. 
We now formalize this statement. 

\begin{lemma}
\lemlab{lem:bicrit:embed:guarantees}
Let $\mu,\nu\in\mathbb{R}^{[\Delta]^d}$ be probability measures so that $\nu$ has support at most $k$ on a set $C$.  
Then with probability at least $0.99$, there exists a probability measure $\nu'$ with support at most $\O{k}$ on the set $\phi(C)$ such that
\[\|W_s(\mu-\nu')\|_1\le\O{d^{1+0.5z}\log^{z-1}\Delta}\cdot\WASSZZ(\mu,\nu).\]
\end{lemma}
\begin{proof}

Let $C=\{C_1,\ldots,C_k\}$ be the support of $\nu$ and let $x\in[\Delta]^d$ be any point in the support of $\mu$. 
Consider each level $t$ for which $x$ and $C_j$ are separated by the quadtree. 
For a fixed axis $j\in[d]$, let $\Proj(x)$ denote the projection of a point $x$ onto the fixed axis and similarly, let $\Proj(C)=\cup_{i\in[k]}\Proj(C_i)$. 
Let $\dist(\Proj(C),\Proj(x))=2^t\cdot\alpha$ for some parameter $\alpha\ge 0$. 

Observe that if $\alpha\ge 1$, then $C_j$ and $x$ will be split at level $t$ along the direction of axis $j$. 
This is because $C_j$ and $x$ will be in different cells regardless of the random choice of the hyperplane chosen along the direction of axis $j$ of the quadtree of length $2^t$. 

On the other hand, if $\alpha=\O{\frac{1}{d\log\Delta}}$, then $C_j$ and $x$ will not be split at level $t$ along the direction of axis $j$. 
This is because we create a point $C'_j$ in $\phi(C)$ corresponding to the projection of $C_j$ onto the hyperplane along the direction of axis $j$ of the quadtree of length $2^t$, so that $C'_j$ and $x$ are not separated by level $j$. 

Otherwise, the probability is $\alpha$ that $C_j$ and $x$ are split at level $t$ along the direction of axis $j$, which causes the contribution to $\|W_s(\mu-\nu')\|_1$ by $\|(2^t\sqrt{d})^z\cdot W_{s,t}(\mu-\nu')\|_1$ to be $(2^t\sqrt{d})^z$. 
Hence, with probability $\alpha$, the distortion between $\dist(C,x)$ is at most $\left(\frac{\sqrt{d}}{\alpha}\right)^z$. 
Therefore, the expected distortion to the contribution of the cost caused by the direction of axis $j$ is at most $\O{d^{0.5z}\log^{z-1}\Delta}$. 
Summing across all $d$ directions, we have that the expected distortion to the contribution of the cost is at most $\O{d^{1+0.5z}\log^{z-1}\Delta}$. 

In summary, the distortion to the cost is at most $\O{d^{1+0.5z}\log^{z-1}\Delta}$ with probability at least $0.9999$ by Markov's inequality. 
Similarly, by \lemref{lem:new:centers:project} and Markov's inequality, it follows that the number of centers in $\nu'$ resulting from $\phi(C)$ is at most $\O{k}$. 
Therefore, the desired claim follows. 
\end{proof}

We define the mapping $\psi(W_s,\mu,\nu)$ for probability measures $\mu,\nu\in\mathbb{R}^{[\Delta]^d}$ so that $\nu$ has support at most $k$ on a set $C$.  
We define the output of the mapping $\nu'=\psi(W_s,\mu,\nu)$ as follows. 
Suppose that mass $m_i$ is moved between $\mu_i$ and $\nu_i$ in the optimal Wasserstein-$z$ transport. 
If there is a separating plane $H$ between $\mu_i$ and $\nu_i$ in $\phi(C)$, then we add mass $m_i$ to $\nu'_i$ at the coordinate that corresponds to the point that is the projection of $\nu_i$ onto $H$. 
Otherwise, we add mass $m_i$ to $\nu'_i$ at the coordinate that corresponds to the point represented by $\nu_i$. 
In other words, $\nu'=\psi(W_s,\mu,\nu)$ is the set $\phi(C)$ with the appropriate mass in the optimal Wasserstein-$z$ transport. 
When the parameters of $W_s$ and $\mu$ are clear from context, we abuse notation and use $\psi(\nu)$ to denote $\psi(W_s,\mu,\nu)$. 

Then from the definition of $\psi$ replacing $\nu'$ in \lemref{lem:bicrit:embed:guarantees}, we immediately have:
\thmquadtreebicrit*

\subsection{\texorpdfstring{$(k,z)$}{(k,z)}-Clustering}
\seclab{sec:kz:alg}
In this section, we give a two-pass dynamic streaming algorithm for $(k,z)$-clustering. 
Recall that for $k$-median, the first pass of our algorithm in \secref{sec:twopass:kmed} was used to set up the EMD embedding, which subsequently gave a $\O{\log k+\log\log\Delta}$ approximation to the sensitivity of each point in the second pass. 
Therefore, we could perform sensitivity sampling through sparse recovery. 
We first describe in \algref{alg:approx:sens:z} the subroutine that can be used to estimate the sensitivity of each point in the stream. 


\begin{algorithm}[!htb]
\caption{Approximation of Sensitivity for $(k,z)$-Clustering}
\alglab{alg:approx:sens:z}
\begin{algorithmic}[1]
\Require{Input dataset $X\subset[\Delta]^d$}
\Ensure{Approximate sensitivity for $x$ for any query $x\in[\Delta]^d$}
\State{$m\gets\tO{kd\log\log\Delta}$}
\State{Initialize a Wasserstein-$z$ sketch, i.e., $m$ instances of $L_1$ sketches on random Wasserstein-$z$ embeddings}
\State{Update the Wasserstein-$z$ sketch with $X$}
\Comment{First pass}
\State{Use the Wasserstein-$z$ sketch to find a near-optimal set $S$ of $k$ centers}
\State{Let $\widetilde{\Cost(X,S)}$ be the estimated cost of $S$ on $X$}
\Comment{\lemref{lem:sens:lower}}
\State{Let $\calM$ be the net from \lemref{lem:kmedian:net}}
\State{$q(x)\gets0$}
\For{$C\in\calM$}
\State{$q(x)\gets\max\left(q(x),\frac{2^z\cdot4\gamma\Cost(x,C)}{\Cost(S,C)+\widetilde{\Cost(X,S)}}\right)$}
\Comment{\lemref{lem:opt:approx:sens}}
\EndFor
\State{\Return $q(x)$}
\end{algorithmic}
\end{algorithm}

Unfortunately, the Wasserstein-$z$ embedding gives a distortion of $\O{d^{1+0.5z}\log^{z-1}\Delta}$, so that sensitivity sampling in the second pass could incur at least an extra $\log^{z-1}\Delta$ factor. 
Even for $z=2$, this no longer gives $\O{\log(n\Delta)}$ total bits of space if we further use $\O{\log(n\Delta)}$ bits of space to represent each sampled point. 
Thus, instead of storing an explicit representation of each point, we store an approximation of its \emph{offset} from each of the centers. 

In particular, for a point $x$ and a set $C'$ of $\O{k}$ centers, let $c'(x)$ be the closest center of $C'$ to $x$. 
We define the offset $y'=x-c'(x)$ and round each coordinate of $y'$ to a power of $\left(1+\poly\left(\eps,\frac{1}{d},\frac{1}{\log\Delta}\right)\right)$ to form a vector $y$. 
Then to represent $y$, it suffices to store the identity of $c'(x)$ as well as the \emph{exponents} of the offsets, which only requires $\O{d\log\frac{\log\Delta}{\eps}}$ bits per sample. 
Hence for $k$-means clustering, the algorithm still uses $o(\log^2(n\Delta))$ total bits of space. 
We give the algorithm in full in \algref{alg:two:pass:dynamic:kz}.

\begin{algorithm}[!htb]
\caption{Two-pass Dynamic Streaming Algorithm for $(k,z)$-Clustering}
\alglab{alg:two:pass:dynamic:kz}
\begin{algorithmic}[1]
\Require{Stream of updates of length $m=\poly(n)$ to coordinates of $[\Delta]^d$, approximation parameter $\eps\in(0,1)$, number of clusters $k$, parameter $z\ge 1$}
\Ensure{$(1+\eps)$-coreset for $(k,z)$-median clustering}
\State{Initialize a Wasserstein-$z$ sketch $W$ for $(k,z)$-clustering}
\State{On the first pass, pass each stream update to $W$}
\State{After the first pass, use $W$ to implicitly define $q(x)$ for all $x\in[\Delta]^d$}
\State{$s\gets\log^{z-1}\Delta\cdot\poly\left(k,d,\frac{1}{\eps^2},\log\log n,\log\log\Delta\right)$}
\State{$p(x)=\min\left(1,\frac{k^2d}{\eps^2}\log k\cdot q(x)\right)$}
\State{Use $G$ to find a set $C'$ of $\O{k}$ centers with cost $\O{2^{z+2}d^{1.5z}\log^{2z}\Delta}\cdot\OPT$}
\State{Initialize an $100s$-sparse recovery algorithm $A$ (see \lemref{lem:sparse:recovery})}
\State{Draw hash function $h$ from a family of hash functions where for all $x\in\mathbb{R}^d$, $\PPr{h(x)=1}=p(x)$ and $h(x)=0$ otherwise}
\For{each update in the second pass}
\If{the update is to $x_t$ and $h(x_t)=1$}
\State{Let $c'(x_t)$ be the closest center of $C'$ to $x_t$}
\State{Let $y'_t$ be the offset $x_t-c'(x_t)$}
\State{Let $y_t$ be $y'_t$ with coordinates rounded to $\left(1+\frac{\eps^{2z}}{8d^{1.5z}\log^{4z}\Delta}\right)^{p(x_t)}$}
\State{Update $A$ with the update to $(y_t,c'(x_t)$ and weight $\frac{1}{q(x)}$}
\EndIf
\EndFor
\State{\Return the weighted points output by $A$}
\end{algorithmic}
\end{algorithm}

We now show that a good bicriteria estimation to the optimal clustering can be used to find a good approximation to the optimal clustering. 
Note this is not immediately evident from the bicriteria estimation because the cost induced by a set of $\O{k}$ centers could be significantly less than the optimal cost induced by a set of $k$ centers. 
\begin{lemma}
\lemlab{lem:bicrit:to:approx}
Let $\alpha>0$ be a fixed parameter and let $\beta$ be a fixed integer. 
Let $B=\{B_1,\ldots,B_{\beta k}\}$ be the centers of an $(\alpha,\beta)$-approximation for $(k,z)$-clustering on a set of points $X=x_1,\ldots,x_n$. 
Let $C=\{C_1,\ldots,C_k\}$ be a $\gamma$-approximation to the weighted set $B$, e.g., each point $B_i$ is given weight corresponding to the number of points in $X$ assigned to $B_i$. 
Then $C_1,\ldots,C_k$ is a $(2^{2z+2}\gamma\alpha)$-approximation to the optimal $(k,z)$-clustering of $X$. 
\end{lemma}
\begin{proof}
Let $\OPT$ be the cost of the optimal $(k,z)$-clustering of $X$. 
Let $V=\{V_1,\ldots,V_k\}$ be an optimal set of $k$ centers for $(k,z)$-clustering of $X$, so that $\Cost(X,V)=\OPT$. 
For each $x\in X$, let $b(x)$ be the center of $B$ assigned to $x$. 
Similarly, let $c(x)$ be the center of $C$ assigned to $x$. 
Then by the generalized triangle inequality, i.e., \factref{fact:triangle},
\begin{align*}
\sum_{x\in X}\dist(x,c(x))^z&\le2^z\left(\sum_{x\in X}\dist(x,b(x))+\sum_{x\in X}\dist(b(x),c(x))\right)\\
&\le2^z\left(\Cost(X,B)+\Cost(B,C))\right)\\
&\le2^z\alpha\OPT+2^z\Cost(B,C),
\end{align*}
since $B$ provides an $(\alpha,\beta)$-bicriteria solution to the optimal $(k,z)$-clustering of $X$, so that $\Cost(X,B)\le\alpha\OPT$. 
Moreover, since $C$ is a $\gamma$-approximation to the optimal set of $k$ centers for $(k,z)$-clustering of the weighted set $B$ and $V$ is a set of $k$ centers, then
\begin{align*}
\sum_{x\in X}\dist(x,c(x))^z&\le2^z\alpha\OPT+2^z\gamma\Cost(B,V)\\
&\le2^z\alpha\OPT+2^{2z}\gamma\left(\Cost(X,V)+\Cost(X,B)\right)\\
&\le2^z\alpha\OPT+2^{2z}\gamma\left(\OPT+\alpha\OPT)\right)\\
&\le2^{2z+1}\gamma(2\alpha\OPT)=2^{2z+2}\gamma\alpha\OPT.
\end{align*}
\end{proof}

We now note that by using \thmref{thm:quadtree:bicrit} and \lemref{lem:bicrit:to:approx}, we can achieve a good approximation to the optimal $(k,z)$-clustering at the end of the first pass of the stream by putting the $L_1$ sketch and the Wasserstein-$z$ embedding together:
\begin{lemma}
\lemlab{lem:sens:lower:z}
There exists an algorithm that uses $\O{kd(\log k+\log\log\Delta}$ words of space and outputs $Z$ such that with probability at least $0.98$, 
\[\OPT\le Z\le\O{d^{1+0.5z}\log^{z-1}\Delta}\OPT,\]
where $\OPT$ is the optimal $(k,z)$-clustering cost at the end of the stream. 
\end{lemma}

We next claim that at the end of the first pass over the stream, we can obtain a bounded approximation to the sensitivity of any query point with respect to the set of points at the end of the stream. 

\begin{lemma}
\lemlab{lem:upper:lower:prob:bounds:z}
There exists an algorithm that uses $\O{kd(\log k+\log\log n)\log\log m}$ words of space and outputs $q(x)$ for all points $x\in\mathbb{R}^d$ in the stream, such that with probability at least $0.98$, 
\[\varphi(x)\le q(x)\le\O{d^{1+0.5z}\log^{z-1}\Delta}\cdot\varphi(x),\]
where $\varphi(x)$ is the sensitivity of $x$ with respect to a point set $X\subset\mathbb{R}^d$ for the $(k,z)$-clustering problem.  
\end{lemma}
\begin{proof}
Let $\widetilde{Z}$ be the smallest estimated cost by any $\O{1}$-approximation to the optimal $(k,z)$-clustering to a set of $\O{k}$ centers induced by the Wasserstein-$z$ sketch, across all sets of $k$ centers in the doubling dimension net. 
Let $\OPT$ be the optimal $(k,z)$-clustering cost of the dataset $X$ at the end of the stream. 
Let $\calE$ be the event that 
\[\widetilde{Z_t}\le\O{d^{1+0.5z}\log^{z-1}\Delta}\cdot\OPT_t,\]
so that by \lemref{lem:sens:lower:z}, we have that $\PPr{\calE}\ge 0.98$. 

Now for any $x\in[\Delta]^d$, the online sensitivity $\varphi(x)$ of $x$ with respect to $X$ is defined by 
\[\varphi(x):=\max_{C:\,C\subset\mathbb{R}^d,\, |C|=k}\frac{\Cost(x,C)}{\Cost(X,C)},\]
where $X$ is the set of points at the end of the stream. 

By \lemref{lem:opt:approx:sens}, a constant factor approximation to a clustering that achieves a clustering cost that is an $\O{d^{1+0.5z}\log^{z-1}\Delta}$-approximation to $\OPT_t$ can be used to compute a $\O{d^{1+0.5z}\log^{z-1}\Delta}$-approximation to $\frac{\Cost(x_t,C)}{\Cost(X_t,C)}$ for any set $C$ of $k$ centers.
Therefore, it follows that conditioned on $\calE$, the estimate $\widetilde{Z_t}$ can be used to compute a $\O{d^{1+0.5z}\log^{z-1}\Delta}$-approximation to $\sigma_t$. 
Hence, the Wasserstein-$z$ sketch can be used to compute $q(x)$ such that with probability at least $0.98$, 
\[\varphi(x)\le q(x)\le\O{d^{1+0.5z}\log^{z-1}\Delta}\cdot\varphi(x),\]
simultaneously for all $x\in[\Delta]^d$.  
In other words, with probability $0.98$, we can guarantee that we have a good approximation for all possible query points $x$.  
\end{proof}

We next recall the following structural property:
\begin{fact}[e.g., Claim 5 in \cite{SohlerW18}]
\factlab{fact:gen:tri}
Let $a,b\ge 0$, $\eps\in(0,1]$, and $p\ge 1$. 
Then
\[(a+b)^p\le(1+\eps)a^p+\left(1+\frac{2p}{\eps}\right)^pb^p.\]
\end{fact}
For each $x_t\in X$, let $f(x_t)=x_t+y_t$, where $y_t$ is defined in \algref{alg:two:pass:dynamic:kz}. 
Let $X'=\{f(x)\}_{x\in X}$. 
\begin{lemma}
\lemlab{lem:x:to:xprime}
With probability at least $0.98$, we have that simultaneously for all $C\subset[\Delta]^d$ with $|C|=k$, 
\[\Cost(C,X')\le(1+\eps)\cdot\Cost(C,X).\]
\end{lemma}
\begin{proof}
Observe that by the triangle inequality, we have
\begin{align*}
\Cost(C,X')&=\sum_{x'\in X'}(\dist(x',C))^z\\
&\le\sum_{x'\in X'}(\dist(x,x')+\dist(x,C))^z.
\end{align*}
By the construction of $x'$, we have that $\dist(x,x')\le\frac{\eps^{2z}}{8d^{1.5z}\log^{4z}\Delta}\cdot\dist(x,C')$. 
Thus, 
\begin{align*}
\Cost(C,X')&\le\sum_{x'\in X'}\left(\frac{\eps^{2z}}{8d^{1.5z}\log^{4z}\Delta}\cdot\dist(x,C')+\dist(x,C)\right)^z.
\end{align*}
By \factref{fact:gen:tri}, we have
\begin{align*}
\Cost(C,X')&\le\sum_{x'\in X'}\left(\left(1+\frac{\eps}{2}\right)(\dist(x,C))^z+\left(1+\frac{4z}{\eps}\right)^z\frac{\eps^{2z}}{(8z)^zd^{1.5z}\log^{4z}\Delta}(\dist(x,C'))^z\right).
\end{align*}
Since $\sum_{x\in X}(\dist(x,C'))^z\le\O{2^{z+2}d^{1.5z}\log^{2z}\Delta}\cdot\min_{C:|C|=k,C\subset[\Delta]^d}\Cost(C,X)$, then it follows that 
\[\Cost(C,X')\le(1+\eps)\cdot\Cost(C,X).\]
Similarly, we have
\begin{align*}
\Cost(C,X)&=\sum_{x\in X}(\dist(x,C))^z\\
&\le\sum_{x\in X}(\dist(x,x')+\dist(x,C))^z.
\end{align*}
Then by the same argument, we have $\Cost(C,X)\le(1+\eps)\Cost(C,X')$. 
\end{proof}

We now claim that with constant probability, we can obtain a bounded approximation to the sensitivity of each point with respect to $X'$.
\begin{lemma}
\lemlab{lem:sens:approx:kz}
There exists an algorithm that uses $\O{(\log k+\log\log n)\log\log m}$ words of space and outputs $q(x)$ for all points $x\in\mathbb{R}^d$ in the stream, such that with probability at least $0.98$, 
\[\varphi(x)\le q(x)\le\O{d^{1+0.5z}\log^{z-1}\Delta}\cdot\varphi(x),\]
where $\varphi(x)$ is the sensitivity of $x$ with respect to the point set $X'$. 
\end{lemma}
\begin{proof}
With probability $0.98$, we have 
\[\Cost(C,X')\le2\Cost(C,X)\le4\Cost(C,X').\]
Thus the sensitivity of $x$ with respect to the point set $X$ is a $2$-approximation to the sensitivity of $x$ with respect to the point set $X'$. 
The desired claim then follows from \lemref{lem:upper:lower:prob:bounds:z}. 
\end{proof}
We recall the following bounds of the sum of the sensitivities for $(k,z)$-clustering.
\begin{theorem}
\thmlab{thm:total:sens}
\cite{VaradarajanX12}
Let $X=x_1,\ldots,x_n\subset[\Delta]^d$ and for each point $x_t$ with $t\in[n]$, let $\sigma_t$ denote its sensitivity for $(k,z)$-clustering for any $z\ge 1$. 
Then
\[\sum_{t=1}^n \sigma_t=\O{2^{2z}k}.\]
\end{theorem}

We now finally justify the guarantees for our two-pass dynamic streaming algorithm for $(k,z)$-clustering. 
\begin{theorem}
There exists a two-pass dynamic streaming algorithm that uses $\tO{\frac{k^3d^{2+0.5z}}{\eps^2}\cdot\log^{z-1}\Delta}$ bits of space and with probability at least $\frac{2}{3}$, outputs a $(1+\eps)$-coreset for $(k,z)$-clustering. 
\end{theorem}
\begin{proof}
Consider \algref{alg:two:pass:dynamic:kz}. 
For each $x_t\in X$, let $f(x_t)=x_t+y_t$, where $y_t$ is defined in \algref{alg:two:pass:dynamic:kz}. 
Let $X'=\{f(x)\}_{x\in X}$. 
For $x_t\in X$, let $\sigma_t$ be the sensitivity of $f(x_t)$ with respect to $X'$ and let $q_t$ be the estimate of $\sigma_t$. 
Let $\calE$ be the event that 
\[\sigma_t\le q_t\le\O{d^{1+0.5z}\log^{z-1}\Delta}\cdot\sigma_t,\]
simultaneously for all $t\in[n]$. 
By \lemref{lem:sens:approx:kz}, we have that $\PPr{\calE}\ge 0.97$. 
Conditioned on $\calE$, $q_t\ge\sigma_t$ for all $t\in[n]$ and thus by \thmref{thm:sens:sample}, the output of \algref{alg:two:pass:dynamic:kz} is a $(1+\O{\eps})$-coreset to $X'$ for $(k,z)$-clustering.
By \lemref{lem:x:to:xprime}, we have that a $(1+\O{\eps})$-coreset for $X'$ is also a $(1+\eps)$-coreset to $X$ for $(k,z)$-clustering. 
Thus the output of \algref{alg:two:pass:dynamic:kz} is a $(1+\eps)$-coreset to $X$ for $(k,z)$-clustering, as desired. 

To analyze the space complexity of \algref{alg:two:pass:dynamic:kz}, we define sequences of indicator random variables $Y_1,\ldots,Y_n$ so that $Y_i=1$ if $x_i$ is sampled by \algref{alg:two:pass:dynamic:kz}. 
Conditioned on $\calE$, we have that $q_t\le\O{d^{1+0.5z}\log^{z-1}\Delta}\cdot\sigma_t$ for all $t\in[n]$. 
Each $x_t$ is sampled with probability $p_t=\min\left(1,\frac{k^2d}{\eps^2}\log k\cdot q_t\right)$, so that
\[\Ex{Y_t}\le\O{\frac{k^2d^{2+0.5z}}{\eps^2}\log k\log^{z-1}\Delta}\cdot\sigma_t.\]
By \thmref{thm:total:sens}, we further have 
\[\sum_{t=1}^n\sigma_t=\O{2^{2z}k}.\]
Thus by Markov's inequality, we have that with probability at least $0.999$,
\[\sum_{t=1}^n Y_t\le\O{\frac{2^{2z}k^3d^{2+0.5z}}{\eps^2}\log k\log^{z-1}\Delta}.\]

Let $\calE'$ be the event that $\O{\frac{2^{2z}k^3d^{2+0.5z}}{\eps^2}\log k\log^{z-1}\Delta}$ points are sampled by \algref{alg:two:pass:dynamic:kz}. 
Conditioned on $\calE'$, algorithm $A$ that is a $100s$-sparse recovery algorithm will recover all of the sampled points, for $s=\O{\frac{2^{2z}k^3d^{2+0.5z}}{\eps^2}\log k\log^{z-1}\Delta}$. 

Crucially, we recall that each of the sampled points in $X'$ are stored as rounded offsets $y_t$ to centers $c'(x_t)$, where the coordinates of $y_t$ have been rounded to a power of $\left(1+\frac{\eps^{2z}}{8d^{1.5z}\log^{4z}\Delta}\right)$. 
Thus, each of the sampled points can be represented in $\O{\log k+z\log\frac{1}{\eps}+z\log d+z\log\log\Delta}$ bits of space. 
By \lemref{lem:sparse:recovery}, the total space used by the sparse recovery algorithm is 
\begin{align*}
\mathcal{O}\bigg(\frac{2^{2z}k^3d^{2+0.5z}}{\eps^2}&\log k\log^{z-1}\Delta\bigg)\cdot\O{\log k+z\log\frac{1}{\eps}+z\log d+z\log\log\Delta}\\
&=\O{\frac{2^{2z}k^3d^{2+0.5z}}{\eps^2}\log k\log^{z-1}\Delta\cdot\left(\log k+z\log\frac{1}{\eps}+z\log d+z\log\log\Delta\right)}
\end{align*}
bits of space.
\end{proof}

\section{Lower Bound and Separations for Dynamic Streams}
\seclab{sec:lb}
In this section, we present a lower bound for $(k,z)$-clustering on insertion-deletion streams. 
We first recall a number of preliminaries from information theory and communication complexity. 

\begin{definition}[Entropy, conditional entropy, mutual information]
Given a pair of random variables $X$ and $Y$ with joint distribution $\rho(x,y)$ and marginal distributions $\rho(x)$ and $\rho(y)$, we define the \emph{entropy} of $X$ by
\[H(X):=-\sum_x \rho(x)\log\rho(x),\] 
the \emph{conditional entropy} by
\[H(X|Y):=-\sum_{x,y} \rho(x,y)\log\frac{\rho(y)}{\rho(x,y)},\]
and the \emph{mutual information} by
\[I(X;Y):=H(X)-H(X|Y)=\sum_{x,y}\rho(x,y)\log\frac{\rho(x,y)}{\rho(x)\rho(y)}.\] 
\end{definition}

\begin{definition}[Information cost]
Let $\Pi$ be a randomized protocol that produces a (possibly random) transcript $\Pi(X_1,\ldots,X_n)$ on inputs $X_1,\ldots,X_n$ drawn from a distribution $\mu$. 
The \emph{information cost} of $\Pi$ with respect to $\mu$ is $I(X_1,\ldots,X_n;\Pi(X_1,\ldots,X_n))$.  
\end{definition}

\begin{definition}[Communication cost and communication complexity]
In a one-way communication protocol $\Pi$, the \emph{communication cost} of $\Pi$ is the maximum bit length of the transcript taken over all inputs and all coin tosses of the protocol. 
The communication complexity of a problem with failure probability $\gamma$ is the minimum communication cost of a protocol that fails with probability at most $\gamma$ over the joint distribution of the inputs and any sources of randomness, e.g., through public or private coins. 
\end{definition}

\begin{fact}
\factlab{fact:ic:cc}
For any distribution $\mu$ and failure probability $\delta\in(0,1)$, the communication cost of any randomized protocol for $\mu$ on a problem $f$ that fails with probability $\delta$ is at least the information cost of $f$ under distribution $\mu$ and failure probability $\delta$. 
\end{fact}

We now define the following distributional augmented equality problem. 

\begin{definition}[Distributional augmented equality problem]
In the distributional augmented equality problem $\DistAugEq_{N,m}$, with probability $\frac{1}{2}$, the input is generated from a YES distribution and with probability $\frac{1}{2}$, the input is generated from a NO distribution. 
In both cases, Alice receives a vector $A\in[N]^m$ and Bob receives a vector $B\in[N]^m$, an index $i\in[m]$ drawn uniformly at random, as well the suffix $A_{i+1},\ldots,A_m$, and their goal is to determine whether $A_i=B_i$ using the minimal amount of communication from Alice to Bob. 
The problem is a one-way communication problem, so that Bob cannot send any information to Alice.
\begin{itemize}
\item
In the YES instance, $A_i=B_i$ is drawn uniformly at random from $[N]$ and all other coordinates of $A$ and $B$ are drawn independently and uniform at random.
\item
In the NO instance, all coordinates of $A$ and $B$ are drawn independently and uniform at random, conditioned on $A_i\neq B_i$. 
\end{itemize}
\end{definition}

\begin{fact}
\factlab{fact:cc:aug:eq}
\cite{JayramW13,MolinaroWY13}
Any protocol for the distributional augmented equality problem $\DistAugEq_{N,m}$ that succeeds with probability at least $1-\frac{1}{\poly(N)}$ requires $\Omega(m\log N)$ information cost. 
\end{fact}

We now show our lower bound for estimating the cost of the optimal $(k,z)$-clustering, e.g., the setting of \cite{CharikarW22}, on dynamic streams. 
Our lower bound works by creating $\poly(n)$ instances of the distributional augmented equality game. 
For each instance, Alice and Bob can view their inputs $A$ and $B$ as strings of length $m$ and alphabet size $N$. 
They can then plant the binary representation of each character in the $\{0,1\}^d$ space, with exponentially increasing weight. 
That is, they plant $\xi^i$ copies of the $i$-th character for some parameter $\xi>1$, where $i\in[m]$. 
For a sufficiently value of $\xi$, the $i$-th copy will dominate all the previous copies. 
Bob can then remove all the points corresponding to characters $i+1,\ldots,m$ and look at the remaining cluster cost to distinguish whether the remaining cost is high, i.e., the $i$-th copy has nonzero contribution, or the remaining cost is low, i.e., the $i$-th copy has no contribution and the cost is from characters $1,\ldots,i-1$, thereby allowing the players to solve the distributional augmented equality problem. 

\thmdynamiclb*
\begin{proof}
Consider an instance of distributional augmented equality where Alice and Bob have input vectors from $[N]^m$, where $N=n^{1/10}$ so that $\log N=\O{\log n}$ and $m=\O{\log n}$. 
Then the vectors $A$ and $B$ can be viewed as $m$ points in $\{0,1\}^{\log N}$. 
In particular, for each $j\in[m]$, we have $A_j,B_j\in[N]$, which can be rewritten in binary with length $\log N$. 

Now Alice and Bob create a stream on $\{0,1\}^d$ with $d=\log N$ as follows. 
For each $j\in[m]$, let $X_j$ be the binary representation of $A_j$ as a $d$-dimensional point. 
Then Alice inserts $(100^z\log^2 N)^i$ copies of the point $X_j$ into the stream, runs a $(k,z)$-clustering algorithm on the stream and passes the state of the algorithm to Bob. 
For each $j\in[m]$, let $Y_j$ be the binary representation of $B_j$ as a $d$-dimensional point. 
Bob takes the state of the algorithm from Alice, inserts $(100^z\log^2 N)^j$ copies of the binary representation of $X_j$ into the stream, and continues running the algorithm on the second portion of the stream.  

To determine whether $A_i=B_i$, Bob then deletes $(100^z\log^2 N)^j$ copies of the binary representation of $A_j$ for $j\in\{i+1,\ldots,m\}$, which are known to Bob. 
Bob then updates the algorithm and queries it for a $(k,z)$-clustering cost with $k=1$. 

Observe that if $A_i\neq B_i$, then at least some coordinate differs in the binary representations of $A_i$ and $B_i$, and so the optimal clustering cost is at least $(100^z\log^2 N)^i\cdot\left(\frac{1}{2}\right)^z$. 
On the other hand if $A_i=B_i$, then the optimal clustering cost is at most 
\[\sum_{j=1}^{i-1}(100^z\log^2 N)^j\cdot(\log N)\le 2(100^z\log^2 N)^{i-1}\cdot(\log N),\]
which is at least a factor of $4$ smaller than $(100^z\log^2 N)^i\cdot\left(\frac{1}{2}\right)^z$, so that any $2$-approximation to the cost of the optimal $(k,z)$-clustering will be able to distinguish whether $A_i=B_i$ or $A_i\neq B_i$, thus solving augmented equality. 


For the dynamic stream to have $\O{n}$ updates, we first set 
\[\sum_{j=1}^m (100^z\log^2 N)^j\le\sqrt{n},\]
which suffices for $\log N=\O{\log n}$ and $m=\O{\log\frac{n}{\log n}}$, assuming constant $z\ge 1$. 
Hence, our construction is a valid dynamic stream with at most $\O{\sqrt{n}}$ updates. 
We then repeat this construction for $I:=\O{\sqrt{n}}$ independent instances drawn from the distributional augmented equality problem, so that the stream has $\O{n}$ updates. 

Since each of the $I=\O{\sqrt{n}}$ instances are independently drawn from the distributional augmented equality problem, then by an averaging argument any algorithm that with probability at least $\frac{2}{3}$, outputs a $2$-approximation to the cost of the optimal $(k,z)$-clustering at all times in the dynamic stream must succeed with probability at least $1-\frac{1}{\Omega(\sqrt{n})}$ on a constant fraction of the distributional augmented equality. 
Specifically, for an index $\ell\in[I]$, we call $\ell$ an \emph{informative} index if the algorithm succeeds on the $\ell$-th instance with probability at least $1-\frac{1}{\Omega(\sqrt{n})}$,  conditioned on the previous outputs being correct. 
Note that since there are $I=\O{\sqrt{n}}$ instances, then there must exist an informative index, since with probability at least $\frac{2}{3}$, the clustering algorithm succeeds at all times. 
Formally, let $\calE_i$ be the event that the $i$-th instance is correct and let $\calE$ be the event that the clustering algorithm succeeds at all times. 
Since the probability of the events $\calE_1,\ldots,\calE_m$ jointly occurring is at least $\frac{2}{3}$, then by expanding 
\[\PPr{\calE} = \prod_{i=1}^{\O{\sqrt{n}}}\PPr{\calE_i\,\mid\calE_1, \ldots, \calE_{i-1}},\]
there exists an informative index $i$, i.e., there exists $i$ for which \[\PPr{\calE_i\,\mid\,\calE_1,\ldots,\calE_{i-1}}\ge 1-\frac{1}{\Omega(\sqrt{n})}.\]
By \factref{fact:cc:aug:eq}, the information cost of the algorithm must be at least $\Omega(\log^2 n)$. 
Thus by \factref{fact:ic:cc}, any algorithm that with probability at least $\frac{2}{3}$, simultaneously outputs a $2$-approximation to the optimal $(k,z)$-clustering cost at all times of a dynamic stream of length $n$ for points in $\{0,1\}^d$ with $d=\Omega(\log n)$ must use $\Omega(\log^2 n)$ bits of space.
\end{proof}

Finally, we remark how our algorithms for $(k,z)$-clustering on insertion-only streams can be further optimized if the goal is simply to output an approximation to the cost of the optimal clustering, rather than a set of $k$ near-optimal centers. 
To that end, recall that the Johnson-Lindenstrauss transform has the following guarantee for $(k,z)$-clustering. 
\begin{theorem}[Dimensionality reduction for $(k,z)$-clustering]
\thmlab{thm:jl:mmr}
\cite{MakarychevMR19}
Let $X\subset\mathbb{R}^d$ be a set of $n$ points, $z\ge 1$ and $m=\O{\frac{z^4}{\eps^2}\log\frac{k}{\eps\delta}}$. 
There exists a family of random linear maps $\pi:\mathbb{R}^d\to\mathbb{R}^m$ such that with probability at least $1-\delta$ over the choice of $\pi$, the cost of every $(k,z)$-clustering of $X$ is preserved up to a $(1+\eps)$-factor under projection by $\pi$. 
\end{theorem}
We now that our algorithms for $(k,z)$-clustering on insertion-only streams can be further optimized if the goal is simply to output an approximation to the cost of the optimal clustering, rather than a set of $k$ near-optimal centers. 
Combined with \thmref{thm:dynamic:lb}, our results show a separation for the $(k,z)$-clustering problem between insertion-only and insertion-deletion streams. 
\thmclustercost*
\begin{proof}
Recall that the algorithm corresponding to \thmref{thm:main} consists of two layers. 
The first layer is using online sensitivity sampling to sample $\poly\left(k,d,\log(nd\Delta),\frac{1}{\eps}\right)$ points into an implicit stream $\calS'$ and the second layer is running a merge-and-reduce algorithm on $\calS'$. 
Now, we project all points to the dimension $m$ in \thmref{thm:jl:mmr}, setting $\delta=\frac{1}{\poly\left(\frac{1}{\eps},\log(nd\Delta)\right)}$. 
We argue that once all points are projected into dimension $m$, online sensitivity sampling still only samples $\poly\left(k,m,\log(nm\Delta),\frac{1}{\eps}\right)$ points and retains correctness at all times in the stream, which shows that $\calS'$ with the projected points is a faithful representation of $\calS$ with the input points. 
Since the merge-and-reduce procedure is only performed on $\calS'$ and thus oblivious to $\calS$, then the correctness follows. 
We can then use the same argument to show that $\calS'$ has $\poly(m,k,\log n,\log\log\Delta)$ points and hence running merge-and-reduce on $\calS'$ uses space $\frac{km}{\eps^2}\polylog(d,k,\log n,\log\Delta)$. 

Since $m=\O{\frac{z^4}{\eps^2}\log\frac{k}{\eps\delta}}$ and $\delta=\frac{1}{\poly\log(nd\Delta)}$, then it follows by \lemref{lem:space:mr} that our algorithm only uses $\tO{\frac{k}{\varepsilon^2}}\cdot(2^{z\log z})\cdot\min\left(\frac{1}{\varepsilon^z},k\right)\cdot\poly(\log\log(n\Delta))$ words of space to output a $(1+\eps)$-approximation to the cost of the optimal $(k,z)$-clustering of an insertion-only stream at all times. 
\end{proof}
Finally, we show a lower bound for any one-pass dynamic streaming algorithm that uses a weighted sample of the input points to estimate the clustering cost. 

Recall that in the AugmentedIndex one-way communication problem, Alice receives a vector $v\in\left[2^t\right]^m$ and Bob receives an index $j\in[m]$, along with the values $v_{j+1},\ldots,v_m$. 
The goal is for Alice to send Bob a message so that with probability at least $\frac{2}{3}$, Bob can successfully compute $v_j$. 
\begin{theorem}[Theorem 10 in~\cite{JayaramW18}]
\thmlab{thm:aug:ind:large:domain}
Any protocol that succeeds for the AugmentedIndex problem with probability at least $\frac{2}{3}$ requires $\Omega(mt)$ communication.
\end{theorem}
Using the AugmentedIndex communication problem, we show that any one-pass dynamic streaming algorithm that uses a weighted sample of the input points to estimate the clustering cost must use $\Omega(\log^2 n)$ bits of space.
\thmonepassdynlb*
\begin{proof}
Suppose there exists a one-pass dynamic streaming algorithm $A$ for $(k,z)$-clustering that uses $o(\log^2 n)$ space. 
Let $C_1=C_2=\frac{1}{1000}$. 
Let $\Delta=n^C$ for some large constant $C>1$ to be fixed. 
Let $\gamma_1$ be an integer parameter, not necessarily independent of $n$, such that the region with distance between $7^{\gamma_1 d}$ and $9^{\gamma_1 d}$ of the origin contains at least $n^{C_1}$ lattice points of $[\Delta]^d$, which is feasible for sufficiently large constant $C$. 
Similarly, let $C$ be large enough so that $\Delta\ge 9^{\gamma_1}\cdot n^{1000}$. 

Suppose Alice and Bob are given an instance of the AugmentedIndex problem with $t=C_1\log n$ and $m=C_2\log n$. 
Alice takes the input vector $v\in\left[2^t\right]^m$ and using public randomness $R$, maps each coordinate $v_i$ with $i\in[m]$ to a lattice point of $[\Delta]^d$ within the region with distance between $7^{\gamma_1 id}$ and $9^{\gamma_1 id}$ of the origin. 
Alice then adds $k-1$ points at $(n^{999},\ldots,n^{999}),\ldots,((k-1)n^{999},\ldots,(k-1)n^{999})$. 
Alice creates a stream $S$ with these points, runs $A$ on $S$ and passes the state of the algorithm to Bob. 

Suppose Bob receives an index $j$ along with the values of $v_{j+1},\ldots,v_m$. 
Bub uses the same public randomness to remove the lattice point of $[\Delta]^d$ within the region with distance between $7^{\gamma_1 id}$ and $9^{\gamma_1 id}$ of the origin, for $i\in[j+1,m]$. 

Let the remaining points form the set $X$ and observe that the optimal clustering for $X$ consists of the points $(n^{999},\ldots,n^{999}),\ldots,((k-1)n^{999},\ldots,(k-1)n^{999})$ as well as the lattice point of $[\Delta]^d$ within the region with distance between $7^{\gamma_1 jd}$ and $9^{\gamma_1 jd}$ of the origin, corresponding to $v_j$. 
Such a clustering has cost at most $\sum_{i=1}^{j-1} 9^{\gamma_1 zid}\le 2\cdot 9^{\gamma_1 z(j-1)d}$. 
On the other hand, any other clustering with no point in the region with distance between $7^{\gamma_1 jd}$ and $9^{\gamma_1 jd}$ of the origin has cost at least $(7^{\gamma_1 zjd}-9^{\gamma_1 z(j+1)d}>4\cdot 9^{\gamma_1 z(j-1)d}$. 

Thus since $A$ is an algorithm that computes the clustering cost from a weighted sample of the input points, then $A$ must include the unique point of $[\Delta]^d$ within the region with distance between $7^{\gamma_1 jd}$ and $9^{\gamma_1 jd}$ of the origin, corresponding to $v_j$. 
Bob can then use this point to recover $v_j$ and solve the Augmented Index problem. 
Thus by \thmref{thm:aug:ind:large:domain}, $A$ must use $\Omega(\log^2 n)$ bits of space. 
\end{proof}

\section*{Acknowledgements}
We thank Peilin Zhong for helpful discussions on the sum of the online sensitivities. 
We also thank anonymous FOCS reviewers for their suggestions to improve the presentation of the paper. 

\def\shortbib{0}
\bibliographystyle{alpha}
\bibliography{references}

\appendix

\section{Quadtree Distortion for \texorpdfstring{$k$}{k}-means Clustering}
\applab{app:quadtree:means:bad}
In this brief appendix, we give a simple example showing that the expected distortion between the squared distances of a set of $n$ points and the estimated distance by a quadtree is $\Omega(n)$.
Consider $\frac{n}{2}$ pairs of points $(a_i,b_i)$ so that $a_i,b_i\in\mathbb{R}$ for $i\in\left[\frac{n}{2}\right]$. 
For each $i\in\left[\frac{n}{2}\right]$, we place $a_i$ uniformly at random in $[\Delta]$ and then we randomly select one of the adjacent lattice points to be $b_i$, so that $|a_i-b_i|=1$. 
Now for a random grid $\calG_j$ with side length $2^j$, the probability that $a_i$ and $b_i$ are separated by $\calG_j$ is at least $\frac{2}{2^j}$. 
If $a_i$ and $b_i$ are separated by $\calG_j$, then the estimated distance by the quadtree is $2^j$, which is a distortion of $2^j$. 
Since there are $\frac{n}{2}$ pairs of points $(a_i,b_i)$, then with probability $0.99$, the estimated distance of some pair points is distorted by at least $\Omega(n)$. 
\end{document}